\documentclass{IEEEtran}

\usepackage{graphicx,color}
\usepackage{cite}
\usepackage{setspace}
\usepackage{amsmath,amsthm,amssymb,bm,booktabs}
\usepackage{multirow}
\usepackage{hhline}
\usepackage{epsfig}
\usepackage{epstopdf}
\usepackage{verbatim}
\usepackage{algorithm}
\usepackage{algpseudocode}
\usepackage{etoolbox}
\usepackage{cases}
\usepackage{array}
\usepackage{subfigure}

\newtheorem{lemma}{Lemma}

\newtheorem{corollary}{Corollary}

\allowdisplaybreaks

\usepackage{forloop}
\newcounter{ct}

\makeatletter 
\def\@eqnnum{{\normalsize \normalcolor (\theequation)}} 
\newenvironment{breakablealgorithm}
  {%
   \begin{center}
     \refstepcounter{algorithm}
     \hrule height.8pt depth0pt \kern2pt
     \renewcommand{\caption}[2][\relax]{%
       {\raggedright\textbf{Algorithm~\thealgorithm.} ##2\par}%
       \ifx\relax##1\relax 
         \addcontentsline{loa}{algorithm}{\protect\numberline{\thealgorithm}##2}%
       \else 
         \addcontentsline{loa}{algorithm}{\protect\numberline{\thealgorithm}##1}%
       \fi
       \kern2pt\hrule\kern2pt
     }
  }{%
     \kern2pt\hrule\relax
   \end{center}
  }
\makeatother

\begin{document}

\title{Vehicular Multistatic OTFS-ISAC: A Geometry- Aware Deployment and Kalman-Based Tracking}
\author{Jyotsna Rani, \textit{Student Member, IEEE}, Kuntal Deka, \textit{Member, IEEE}, Ganesh Prasad, \textit{Member, IEEE} and Zilong Liu, \textit{Senior Member, IEEE}.

	\thanks{J. Rani and K. Deka are with the Department of Electronics \& Electrical Engineering, Indian Institute of Technology Guwahati, India (e-mail: \{r.jyotsna, kuntaldeka\}@iitg.ac.in).
		
	G. Prasad is with the Department of Electronics and Communication Engineering, National Institute of Technology Silchar, India (e-mail: gpkeshri@ece.nits.ac.in).

     Z. Liu is with the School of Computer Science and Electronics Engineering, University of Essex, Colchester, UK (e-mail: zilong.liu@essex.ac.uk)
    
    }}

\maketitle
\begin{abstract}
Integrated sensing and communication (ISAC) is a promising paradigm for next-generation vehicular networks, yet existing orthogonal frequency-division multiplexing (OFDM)-based designs suffer from limited spatial diversity and severe sensitivity to Doppler and multipath effects. While orthogonal time-frequency space (OTFS) modulation offers robustness under high mobility, the impact of spatial node deployment in multistatic OTFS-ISAC has remained largely unexplored. This paper presents the first geometry-aware multistatic OTFS-ISAC framework, in which a triangulation-based cooperative sensing approach is developed for joint target localization and velocity estimation. Closed-form expressions for the localization error covariance are derived under general receiver topologies, revealing that maximizing the triangulation area is fundamental to minimizing estimation error. This leads to a near-optimal deployment strategy based on orthogonal receiver placement and its equivalence to multi-antenna architectures with cubic-order error reduction. To enable reliable tracking of moving targets, a correlated random walk (CRW)-based Kalman filter (KF) framework is integrated into multistatic OTFS-ISAC for active sensing and ISAC. Numerical results demonstrate significant reductions in localization root-mean-square error (RMSE) and communication bit error rate (BER), highlighting the effectiveness of geometry-aware, KF-assisted multistatic OTFS-ISAC in dynamic vehicular environments.
\end{abstract}
   \vspace{-2mm}

\begin{IEEEkeywords}
	Integrated sensing and communication, orthogonal time-frequency space modulation, distributed networks, Kalman filtering.
\end{IEEEkeywords}
\vspace{-2mm}

\section{Introduction}
The advancement of vehicular networks toward intelligent transportation is driven by vehicle-to-everything (V2X) communication, which enables information exchange among vehicles, infrastructure, and pedestrians. This evolution has created demands for wireless technologies capable of delivering communication and real-time sensing \cite{noor24}. Integrated sensing and communication (ISAC) has emerged as a solution by unifying hardware and spectrum resources for simultaneous data transmission and environmental sensing. This joint design not only enhances spectral-efficiency but also simplifies the system implementation \cite{FLiuITC2020JRCDASOARA}. The  orthogonal frequency division multiplexing (OFDM) waveform is sensitive to Doppler and timing synchronization error in high-mobility channels and hence may not be capable. On the other hand, orthogonal time frequency space (OTFS) transmits information symbols into the delay-Doppler (DD) domain, whereby the effective channel matrix is quasi-static and sparse. Because of this, OTFS is well-suited for ISAC in vehicular networks as it can support high-resolution sensing and communication performances \cite{RHadaniWCNC2017OTFS,yua14}. However, the main challenges for OTFS-based ISAC are tracking moving vehicles, keeping communication reliable under strong Doppler, estimating speeds of multiple targets, and achieving accurate localization in multipath conditions under limited resources. Large co-located antenna arrays can provide high localization accuracy but incur substantial hardware and computational costs due to high-dimensional DD domain processing, multi-dimensional fast Fourier transform (FFT) operations, and exhaustive DD parameter estimation. Multistatic ISAC overcomes these issues by coordinating spatially separated low-complexity nodes, offering broad coverage, cooperative sensing, and robust communication while keeping per-node complexity low \cite{han04}. In this work, we consider multistatic OTFS-based ISAC along with Kalman filtering (KF) to enhance sensing and communication performance.

\vspace{-1mm}
\subsection{Related Works} 
\subsubsection{OTFS-Based ISAC}
A major advantage of OTFS is that it can capture the channel dynamics caused by scatterer geometry to provide robustness against Doppler impairments and multipath fading in high-mobility settings such as vehicular networks, compared to OFDM. In \cite{yua14}, the quasi-static DD channels were used to infer communication states from sensing parameters, achieving reliable transmissions under dynamic conditions. Further, the authors in \cite{wu19} advanced this study by formulating a DD-domain precoder that jointly minimizes BER under the constraint of transmit power and Cramér–Rao Bound (CRB). In \cite{zac16}, practical resolution limits were addressed by introducing a low-complexity orthogonal matching pursuit with fractional refinement (OMPFR) algorithm for fractional DD estimation. Precise target profiling in monostatic OTFS-ISAC were enabled without excessive complexity. When extended to multistatic settings, OTFS can unlock greater performance gains by combining spatial diversity with DD-domain processing. In \cite{fan17}, power allocation strategies for cell-free massive multiple-input–multiple-output (CF-mMIMO) OTFS-ISAC were designed for significant SINR and spectral efficiency improvements over OFDM. Further, authors in \cite{sruti18} developed sparse recovery-based association and localization methods for multistatic OTFS-ISAC, yielding scalable and bandwidth-efficient solutions for vehicular scenarios. On contrary, the work in \cite{zeg25} leveraged OTFS pilot sidelobes to achieve super-resolution range-velocity profiling, surpassing conventional time-bandwidth limits in automotive radar. Thus, the OTFS-enabled multistatic ISAC achieves robust and scalable integration of communication and sensing.

\vspace{-0.5mm}
\subsubsection{Monostatic ISAC}
Monostatic ISAC systems operate through co-located transmitter and receiver by utilizing a clock for synchronization and knowledge of the transmitted signals, thereby eliminating coordination and hardware overhead requried by the multi-node systems~\cite{su05,de06}. 
Many works have applied OFDM-based monostatic ISAC together with MIMO arrays, showing that such systems can provide range, angle, and velocity estimation~\cite{su05,zh08}. 
In vehicular environments, monostatic ISAC is desired since it is compact, cost-efficient, and compatible with V2X standards~\cite{de06,che07}. But, the need for full duplex technology to handle self-interference, monostatic systems rely on one observation point only, thus suffering from compromised spatial diversity and robustness. This limitation becomes more critical in high-mobility environments, where Doppler and multipath effects make reliable target tracking difficult to achieve.

\vspace{-0.5mm}
\subsubsection{Multistatic ISAC}
Multistatic ISAC has emerged as an approach to overcome inherent limitations of monostatic systems by allowing multiple base stations, user devices, or vehicles to cooperate for ISAC. This cooperation increases spatial diversity and enhances localization accuracy, which are important in high-mobility environments. Earlier works on uplink cooperative ISAC focused on waveform and resource allocation, showing that joint optimization could reduce the Cramér--Rao lower bound (CRLB) for sensing while ensuring communication quality~\cite{li09}. In vehicular networks, cooperative time-division ISAC frameworks demonstrated that sharing sensing data among automated vehicles reduces detection errors and latency~\cite{zh10}. At the network level, multistatic ISAC schemes based on coordinated multi-point (CoMP) transmission and multi-static radar reported scaling gains in sensing accuracy as the number of cooperating transceivers increases~\cite{hon11,men12}. Mutual information-based sensing metrics were introduced to unify communication-sensing trade-offs, effectively modeling multistatic ISAC as a virtual MIMO system~\cite{yin13}. Besides, KF was employed in \cite{sag26} to recursively estimate target states by fusing noisy position--velocity measurements with an assumed motion model. In multistatic OTFS-ISAC, this approach was essential for mitigating mobility- and geometry-induced uncertainties, thereby improving robustness of moving-target sensing and reliability of the associated communications. 


\vspace{-2mm}
\subsection{Motivation and Key Contributions}  
 Majority of the multistatic OTFS-ISAC studies have focused  on waveform design rather than node deployment. Since deployment geometry affects localization accuracy, sensing coverage, and communication reliability, optimal placement remains an open problem. Moreover, integration of KF with geometry-aware multistatic OTFS-ISAC for tracking moving targets is unexplored. The key contributions of the work are:
\begin{itemize}
    \item We introduce topology of distributed ISAC framework with  OTFS transform waveform and receive signal processing. We then present methods for active sensing and ISAC to estimate target range, radial velocity, and achievable data rate at receivers. Based on system topology, we propose a geometrical triangulation approach that exploits spatial diversity to improve accuracy of target localization and velocity estimation. Using this spatial diversity, we describe two methods to improve estimation of targets' location and velocity.
    \item Next, we describe a lemma that  highlights role of triangulation area in reducing localization error and show optimizing this area is essential for an accurate estimation. A corollary demonstrates that when receivers are placed along orthogonal axes, estimation error can be minimized by maximizing the triangulation area. We further derive closed-form expressions for trace and the maximum eigenvalue of covariance matrix of the estimated target locations under a general topology with randomly deployed receivers. Using these results, we prove in a subsequent lemma that a suboptimal yet effective topology is obtained when receivers are positioned along orthogonal axes, which is equivalent to a configuration of three nodes equipped with multiple independent antennas. Additionally, we show that this equivalent topology reduces the estimation error by a factor proportional to the cubic power of number of independent antennas per node.
    \item For moving vehicles, we employ KF to improve accuracy of state estimation. The prediction model is derived from a correlated random walk (CRW), which is converted to discrete-time domain for tractable implementation. Based on this, we present a KF-based algorithm that estimates states of the targets, specifically their position and velocity. We extend application of KF to develop a second algorithm aimed at enhancing active sensing performance by refining prediction and correction steps. Furthermore, we propose a third algorithm that leverages KF for joint passive sensing and communication, thereby integrating tracking and information exchange within a unified scheme. An analysis of time complexity of proposed algorithms is provided.

    \item Finally, proposed geometry-aware multistatic OTFS-ISAC framework is validated through extensive numerical evaluations. The results show that cooperative triangulation substantially reduces localization root mean square error (RMSE) compared with monostatic and randomly deployed multistatic schemes, with orthogonal receiver placement achieving minimum error by maximizing effective triangulation area, in agreement with derived error covariance analysis. The integration of a CRW-based KF further improves both localization and velocity estimation accuracy for moving targets, yielding faster convergence and lower steady-state RMSE in active and passive sensing modes. In joint passive sensing and communication scenario, the KF-assisted sensing-aided channel reconstruction leads to improved BER performance relative to pilot-only ISAC schemes. Comparative simulations across multiple benchmark configurations clearly illustrate trade-offs between deployment geometry, sensing accuracy, and communication reliability. These results consistently demonstrate performance gains of proposed geometry-optimized, KF-assisted multistatic OTFS-ISAC framework.

\end{itemize}
\vspace{-1mm}
\subsubsection*{Notations}
In this paper, bold lowercase and uppercase letters denote vectors and matrices, respectively. $\mathbb{C}^{M\times N}$ denotes the set of $M \times N$ complex-valued matrices, and $\mathbb{E}\{\cdot\}$ denotes statistical expectation. The operators $\otimes$, $\mathrm{vec}(\cdot)$, and $\mathrm{vec}^{-1}(\cdot)$ denote the Kronecker product, vectorization, and inverse vectorization, respectively. The absolute value and Frobenius norm are denoted by $|\cdot|$ and $\lVert\cdot\rVert$, while $\mathrm{card}(\cdot)$ gives set cardinality. The trace, Hermitian transpose, diagonal, covariance, and maximum eigenvalue operators are denoted by $\mathrm{tr}(\cdot)$, $(\cdot)^H$, $\mathrm{diag}\{\cdot\}$, $\mathrm{cov}(\cdot,\cdot)$, and $\kappa_{\max}(\cdot)$, respectively. $\mathcal{O}(\cdot)$ characterizes asymptotic computational complexity.

\vspace{-2mm}
\section{System Architecture}
\begin{figure}[!t]
\centering
\includegraphics[width=2.5in]{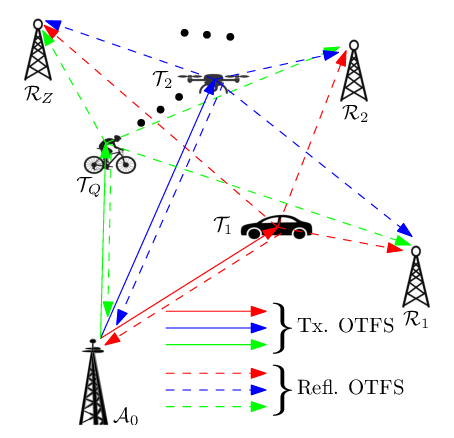}
\caption{Network topology for multistatic ISAC.}
\label{fig:ISAC_OTFS_Arch}
\end{figure}

\vspace{-1mm}
\subsection{Topology}
As shown in Fig.~\ref{fig:ISAC_OTFS_Arch}, the system uses a multistatic ISAC setup which includes one anchor node (AN), called $\mathcal{A}_0$, and several receiver nodes $\{\mathcal{R}_1, \ldots, \mathcal{R}_Z\}$ placed in different locations\footnote{A coplanar (2D) model is adopted as vehicular nodes and targets lie primarily on the road plane, rendering elevation effects negligible; the framework readily extends to 3D by incorporating elevation without altering the core signal processing or estimation methodology.}. The system works in two modes: active and passive sensing. In both modes, it uses monostatic and bistatic sensing. In monostatic sensing, the anchor node $\mathcal{A}_0$ senses the environment using the echoes of its own signals reflected from the targets $\{\mathcal{T}_1, \ldots, \mathcal{T}_Q\}$. In bistatic sensing, the anchor sends signals, and the receiver node $\{\mathcal{R}_z$; $z\in\{1,\cdots,Z\}\}$ collects the echoes reflected from the targets. This allows the anchor and receivers to work together as pairs. In the active sensing mode, the anchor transmits known probing waveforms, which are also available at the receiver nodes, thereby enabling cooperative bistatic sensing. In the passive sensing mode, the anchor transmits standard communication signals embedded with pilots, and the receivers exploit these signals for simultaneous information decoding and environmental sensing without emitting any waveform for sensing. Thus, it reduces the energy consumption and minimizes the interference. Since the propagation environment is doubly spread in delay and Doppler, the transmit signal is generated using OTFS modulation, as described next.

\vspace{-1mm}
\subsection{Transmit OTFS Signal}
In OTFS signaling, information data sequence, generated from quadrature amplitude modulation (QAM) alphabet $\mathbb{A} \in \mathbb{C}$, represented by vector $\mathbf{d}_I \in \mathbb{A}^{MN}$ is reshaped into a DD domain matrix $\mathbf{D}_I \in \mathbb{A}^{M \times N}$. Transformation is carried out by dividing $\mathbf{d}_I$ into $N$ consecutive segments of length $M$, with each segment forming one column of matrix $\mathbf{D}_I$. The parameters $M$ and $N$ correspond to number of delay and Doppler bins, respectively. To facilitate ISAC, a pilot matrix $\mathbf{D}_P \in \mathbb{C}^{M \times N}$ is introduced along with data matrix. The indices for delay and Doppler bins in both $\mathbf{D}_I$ and $\mathbf{D}_P$ are denoted by $l \in \{0,\cdots, M-1\}$ and $k \in \{0, \cdots, N-1\}$, respectively. Note that total bandwidth and duration of OTFS signal is $M\Delta f$ and $NT$. Here, $\Delta f$ denotes subcarrier spacing in time-frequency grid, $T$ is the OTFS symbol duration, and $T\Delta f = 1$. Therefore, resolutions of delay and Doppler shift are $1/M\Delta f$ and $1/NT$. Further, indices $l$ and $k$ correspond to delay $l/M\Delta f$ and Doppler shift $k/NT$. The pilot matrix $\mathbf{D}_P$ is designed to have a single nonzero entry at location $(l_p, k_p)$, i.e., $\mathbf{D}_P(l,k) = \sqrt{MN \sigma_P^2}$ if $(l, k) = (l_p, k_p)$, and zero elsewhere. This ensures that average pilot power is maintained as $\mathbb{E}\{|\mathbf{D}_P(l,k)|^2\} = \sigma_P^2$. To further enhance reliability of pilot-based sensing, a guard band is allocated around pilot position in delay--Doppler domain to suppress interference caused by channel's maximum delay and Doppler spread. Such an isolation region prevents leakage from surrounding data symbols, especially under practical pulse-shaping and fractional Doppler effects, thereby preserving pilot’s orthogonality and improving estimation accuracy in ISAC operation. The resultant data ISAC is given by $\mathbf{D} = \mathbf{D}_I + \mathbf{D}_P$. Next, DD data is converted into time-frequency domain using an operation called inverse symplectic finite Fourier transform (ISFFT). This is done by multiplying data with two Fourier transform matrices: $\mathbf{X} = \mathbf{F}_M \, \mathbf{D} \, \mathbf{F}_N^H$, where $\mathbf{F}_M$ and $\mathbf{F}_N$ are standard FFT matrices. Thereafter, time-frequency signal is converted to time domain using an inverse FFT, and then passed through a pulse-shaping filter $p(t)$ of duration $T$ which is denoted in diagonal matrix form as $\mathbf{P}_{\text{tx}} = \text{diag}\{p_{\text{tx}}(0), p_{\text{tx}}(1),\cdots,p_{\text{tx}}(M-1)\}$ using its samples. This filter helps to shape signal for transmission through wireless channel. Thus OTFS transmit signal $\mathbf{S}$ in matrix form is

\begin{align}
    \mathbf{S} = \mathbf{P}_\text{tx} \, \mathbf{F}_M^H \left( \mathbf{F}_M \, \mathbf{D} \, \mathbf{F}_N^H \right).
\end{align}
Further, the resulting time-domain signal vector $\mathbf{s}\in \mathbb{C}^{MN \times 1}$ is obtained by stacking the columns of $\mathbf{S}$ as
\begin{align}\label{eq:d2s}
    \mathbf{s} = \left( \mathbf{F}_N^H \otimes \mathbf{P}_\text{tx} \right) \mathbf{d},
\end{align}
where the vector $\mathbf{d}=\mathrm{vec}(\mathbf{D}) \in \mathbb{C}^{MN \times 1}$.

\vspace{-2mm}
\subsection{Received Signal Processing}
In a doubly-dispersive channel, the received signal $\mathbf{r}_j$ at the $j$th receiver is the superposition of the signals via the $P$ paths, where each path introduces a gain, cyclic delay, and Doppler shift. The channel matrix $\mathbf{H}_j$ for the $j$th receiver is given by
\begin{align}\label{eq:H_ch}
    \mathbf{H}_j = \sum_{q=1}^P h_{q,j} \, \mathbf{\Pi}^{l_{q,j}}\mathbf{\Delta}^{k_{q,j}}.
\end{align}
Here, $h_{q,j}$ is the channel coefficient and $l_{q,j}$, $k_{q,j}$ are the indices of the delay and Doppler shift of the $q$th path. $\mathbf{\Pi}^{l_{q,j}}$ is $l_{q,j}$-step forward cyclic shift to the permutation matrix $\mathbf{\Pi}$ as given in \cite[eq. (9)]{rav01}. This parameter gives $l_{q,j}$-step cyclic shift to the transmit signal $\mathbf{s}$ in the $q$th path. On the other hand $\mathbf{\Delta}^{k_{q,j}}$ gives the Doppler shift by frequency $2\pi k_{q,j}/MN$, where $\mathbf{\Delta}=\mathrm{diag}\{c^0,c^1,\cdots,c^{MN-1}\}$, where $c=\mathrm{e}^{i2\pi/MN}$. Here both the matrices $\mathbf{\Pi}$ and $\mathbf{\Delta}$ are of size $MN\times MN$. Using \eqref{eq:H_ch}, the received signal $\mathbf{r}_j$ is expressed as
\begin{align}\label{eq:r_j}
    \mathbf{r}_j = \mathbf{H}_j\mathbf{s} + \mathbf{n}_j,
\end{align}
where $\mathbf{n}_j\in \mathcal{CN}(\mathbf{0},\sigma_j^2\mathbf{I}_{MN})$. In order to perform the sensing and communication using the received signal $\mathbf{r}_j$, we first convert it from the time domain to DD domain using the inverse process of obtaining $\mathbf{s}$ from $\mathbf{d}$ (cf. \eqref{eq:d2s}), which is described as follows. The time domain signal in matrix form is obtained as: $\mathbf{R}_j=\mathrm{vec}^{-1}(\mathbf{r}_j)$. Then, it is passed through a receiver filter $\mathbf{P}_\text{rx}=\mathrm{diag}\{p_\text{rx}(0), p_\text{rx}(1),\cdots,p_\text{rx}(M-1)\}$ which is generated using sampling of pulse-shaping waveform $p_\text{rx}(t)$ at the receiver. Then, delay-Doppler domain signal $\mathbf{Y}_j$ is obtained by performing $M$-point FFT followed by SFFT on it which is represented as: $\mathbf{Y}_j = \mathbf{F}_M^H(\mathbf{F}_M \mathbf{P}_{rx} \mathbf{R}_j)\mathbf{F}_N=\mathbf{P}_{rx} \mathbf{R}_j\mathbf{F}_N$. The corresponding vector $\mathbf{y}_j$ is obtained as
\begin{subequations}
\begin{align}\label{eq:y_a}
    \mathbf{y}_j& = (\mathbf{F}_N \otimes \mathbf{P}_{rx})\mathbf{r}_j,\\\label{eq:y_b}
                & = \overline{\mathbf{H}}_j\mathbf{d} + \overline{\mathbf{n}}_j,
\end{align}
\end{subequations}
where \eqref{eq:y_b} is obtained from \eqref{eq:y_a} by substituting $\mathbf{r}_j$, $\mathbf{s}$, and $\mathbf{H}_j$ from \eqref{eq:r_j}, \eqref{eq:d2s}, and \eqref{eq:H_ch}. Thus, $\overline{\mathbf{H}}_j = \sum_{q=1}^P h_{q,j}\mathcal{T}(l_{q,j},k_{q,j})$ and $\mathcal{T}(l_{q,j},k_{q,j}) = (\mathbf{F}_N \otimes \mathbf{P}_{rx}) \mathbf{\Pi}^{l_{q,j}}\mathbf{\Delta}^{k_{q,j}} (\mathbf{F}_N^H \otimes \mathbf{P}_{tx})$. Note that $\mathbf{P}_{tx} = \mathbf{P}_{rx} = \mathbf{I}_M$ for rectangular pulse-shaping filter. The noise $\overline{\mathbf{n}}_j = (\mathbf{F}_N \otimes \mathbf{P}_{rx})\mathbf{n}_j$ and it can be shown that it follows the Gaussian distribution $\overline{\mathbf{n}}_j\in \mathcal{CN}(\mathbf{0},\sigma_j^2\mathbf{I}_{MN})$ for the normalized unitary matrix $\mathbf{F}_N$.  

\vspace{-2mm}
\subsection{Active Sensing Procedure}\label{sec:act_sense}
In active sensing, the transmit data vector $\mathbf{d}$ is assumed to be known at all receiving nodes \cite{gau20}. Utilizing this prior knowledge, the targets' range and radial velocity relative to the $j$th receiver can be determined by estimating the set of parameters $\bm{h}_{j}$, $\bm{\tau}_j$, and $\bm{\nu}_j$, which represent the complex gains, propagation delays, and Doppler shifts, respectively. These parameters are obtained by solving a maximum likelihood estimation problem. Under the assumption of additive Gaussian noise at the receivers, the log-likelihood function simplifies to a least-squares minimization problem, where the received delay--Doppler vector $\mathbf{y}_j$ is compared with its expected value $\overline{\mathbf{H}}_j\mathbf{d}$. This leads to the following estimator formulation:
\begin{align}\label{eq:est_min}
(\widehat{\mathbf{h}}_j,\widehat{\bm{\tau}}_j,\widehat{\bm{\nu}}_j) = 
\arg\min_{(\mathbf{h}_j,\bm{\tau}_j,\bm{\nu}_j)} 
\left\lVert 
\mathbf{y}_{j}-\sum_{q=1}^P h_{q,j}\mathcal{T}(\tau_{q,j},\nu_{q,j})\mathbf{d} 
\right\rVert_{F}^2
\end{align}
where,  $\mathcal{T}(\tau_{q,j},\nu_{q,j})$ denotes the combined delay and Doppler shift operator acting on the transmit signal. The problem in \eqref{eq:est_min} is nonconvex due to the discrete nature of the delay and Doppler parameters and is computationally intensive because the search space spans the high-dimensional domain $\mathbb{C}^P \times \mathbb{R}_+^P \times \mathbb{R}^P$. To mitigate this complexity, we adopt a sequential estimation strategy. In the first step, the parameters of the dominant path delay $\tau_{1,j}$ and Doppler shift $\nu_{1,j}$ are obtained by maximizing the correlation between the received signal and the delayed--Doppler shifted transmit waveform:
\begin{align}\label{eq:est_first}
(\widehat{\tau}_{1,j},\widehat{\nu}_{1,j}) 
= \arg\max_{(\tau_{1,j},\nu_{1,j})} 
\left| 
\left(\mathcal{T}(\tau_{1,j},\nu_{1,j})\mathbf{d}\right)^H \mathbf{y}_j 
\right|^2.
\end{align}
Given $(\widehat{\tau}_{1,j},\widehat{\nu}_{1,j})$, gain $h_{1,j}$ is estimated by minimizing residual error, which leads to
\begin{align}\label{eq:first_gain}
\widehat{h}_{1,j} = 
\frac{
\left(\mathcal{T}(\widehat{\tau}_{1,j},\widehat{\nu}_{1,j})\mathbf{d}\right)^H \mathbf{y}_j
}{
\mathbf{d}^H 
\mathcal{T}(\widehat{\tau}_{1,j},\widehat{\nu}_{1,j})^H 
\mathcal{T}(\widehat{\tau}_{1,j},\widehat{\nu}_{1,j}) 
\mathbf{d}
}.
\end{align}
To ensure reliable detection of subsequent reflections, the interference generated by previously detected paths must be suppressed before estimating the remaining components. This interference cancellation step substantially improves the estimator's robustness, particularly in multipath-rich vehicular environments, by preventing strong paths from masking weaker ones. Subsequent paths are estimated successively. For the estimation of the $i$th path parameters, the contribution from previously estimated paths is subtracted from the received signal, and the same correlation-maximization and gain estimation procedure is applied as:
\begin{subequations}\label{eq:ith_est}
    \begin{align}
        (\widehat{\tau}_{i,j},\widehat{\nu}_{i,j}) &= \arg\max_{(\tau_{i,j},\nu_{i,j})}\left| (\mathcal{T}(\tau_{i,j},\nu_{i,j})\mathbf{d})^H\left(\mathbf{y}_j \right.\right.\\\nonumber
        &\;\;\left.\left.-\sum_{q=1}^{i-1}\widehat{h}_{q,j}\mathcal{T}(\widehat{\tau}_{q,j},\widehat{\nu}_{q,j})\mathbf{d}\right) \right|^2,\\\label{eq:h_est}
        \widehat{h}_{i,j} &= \frac{(\mathcal{T}(\widehat{\tau}_{i,j},\widehat{\nu}_{i,j})\mathbf{d})^H\mathbf{y}_j}{\mathbf{d}^H\mathcal{T}(\widehat{\tau}_{i,j},\widehat{\nu}_{i,j})^H\mathcal{T}(\widehat{\tau}_{i,j},\widehat{\nu}_{i,j})\mathbf{d}}.
    \end{align}
\end{subequations}
\subsubsection*{Complexity}
The computational cost of the sequential estimator is dominated by the two-dimensional delay-Doppler correlation search. For each delay-Doppler pair, the correlation involves an inner product of length $MN$ that gives the complexity of $\mathcal{O}(G_\tau G_\nu MN)$ per path at a single receiver, where $G_\tau$ and $G_\nu$ denote the grid sizes in delay and Doppler, respectively. The gain estimation step requires only a few additional inner products of length $MN$, i.e., $\mathcal{O}(MN)$, and is negligible compared with the search stage. Repeating the procedure for $P$ paths results in a per-receiver complexity of $\mathcal{O}(P G_\tau G_\nu MN)$. Considering all $Z+1$ receivers, the overall complexity is $\mathcal{O}((Z+1)P G_\tau G_\nu MN)$. Thus, the overall complexity grows linearly with the number of paths, the signal dimension, and the number of receivers, but quadratically with the resolution of the delay-Doppler grid, with the correlation search being the dominant factor.

\vspace{-0.5mm}
\subsubsection*{Computation of Range and Velocity}
	Finally, the self-reflected signal is received by receiver $\mathcal{A}_0$, whereby the range $\rho_{i,0}$ and the radial velocity $v_{i,0}$ of the $i$th target can be computed as in \eqref{eq:est_self} below. For receiver $\mathcal{R}_j$, the range $\rho_{i,j}$ and radial velocity $v_{i,j}$ of the $i$th target are given in \eqref{eq:est_ref}.
	\begin{subequations}\label{eq:est_tar_para}
		\begin{align}\label{eq:est_self}
			&\widehat{\rho}_{i,0} = \widehat{\tau}_{i,0}c/2;\;\; \widehat{v}_{i,0} = \widehat{\nu}_{i,0}c/2f_c,\\\label{eq:est_ref}
			& \widehat{\rho}_{i,j} = \widehat{\tau}_{i,j}c-\widehat{\rho}_{i,0}; \;\; \widehat{v}_{i,j} = (\widehat{\nu}_{i,j} - \widehat{\nu}_{i,0})c/f_c.
		\end{align}
	\end{subequations}

\vspace{-2mm}
\subsection{ISAC Procedure}
In ISAC, only the pilot signal is assumed to be known at all receivers. Consequently, the channel, characterized by its gain, delay, and Doppler shift components, is first estimated. Based on this initial channel estimate, the information data is subsequently detected. Utilizing both the estimated data and the known pilot signal, the channel is re-estimated in the next iteration. This iterative process continues until convergence is achieved. For an initial (coarse) estimate of the channel, the delay-Doppler matrix of the pilot signal, denoted as $\mathbf{D}_P$, is vectorized as $\mathbf{d}_P = \text{vec}(\mathbf{D}_P)$. Using it, each receiver estimates the channel while treating the transmitted information data $\mathbf{d}_I$ as interference. Analogous to \eqref{eq:est_min}, the channel coarse estimator at the $j$-th receiver is expressed as
\begin{align}\label{eq:est_pilot_min}(\widehat{\mathbf{h}}_j^{(0)},\widehat{\bm{\tau}}_j^{(0)},\widehat{\bm{\nu}}_j^{(0)})\hspace{-0.5mm} = \hspace{-0.5mm}\arg\hspace{-0.5mm}\min_{(\mathbf{h}_j,\bm{\tau}_j,\bm{\nu}_j)}\hspace{-0.5mm} \left\lVert \mathbf{y}_{j}\hspace{-0.5mm}-\hspace{-0.5mm}\sum_{q=1}^P\hspace{-0.5mm} h_{q,j}\hspace{-0.5mm}\mathcal{T}(\tau_{q,j},\hspace{-0.5mm}\nu_{q,j})\mathbf{d}_P \right\rVert_{F}^2,
\end{align}
The problem in \eqref{eq:est_pilot_min} can be solved using the sequential estimation strategy as described in \eqref{eq:est_first}, \eqref{eq:first_gain}, and \eqref{eq:ith_est}. However, this initial coarse estimation is inherently less accurate than in the active sensing scenario, where both the pilot $\mathbf{d}_P$ and the information $\mathbf{d}_I$ are known to all receivers. To improve estimation accuracy in this passive setting, we adopt an iterative refinement algorithm as outlined below. Starting with the coarse channel estimate $(\widehat{\mathbf{h}}_j^{(0)}, \widehat{\bm{\tau}}_j^{(0)}, \widehat{\bm{\nu}}_j^{(0)})$, the signal $\mathbf{d}$ is recovered using a channel equalizer. This step is formulated as a regularized least-squares optimization problem:
\begin{align}\label{eq:reg_data_est}
    \min_{\mathbf{d}} \left\lVert \mathbf{y}_j - \overline{\mathbf{H}}_j\mathbf{d} \right\rVert_{F}^2 + \mu \left\lVert \mathbf{d} \right\rVert_{F}^2,
\end{align}
where $\mu$ is a regularization parameter that ensures numerical stability, especially when the estimated channel matrix $\overline{\mathbf{H}}_j$ is ill-conditioned or rank-deficient. Without regularization, a nearly singular $\overline{\mathbf{H}}_j$ can lead to solutions that are overly sensitive to noise, resulting in overfitting. The regularization term $\mu \lVert \mathbf{d} \rVert^2$ mitigates this effect by penalizing large values in the solution, effectively reducing the bias toward noise-induced errors. Based on the minimum mean square error (MMSE), the closed-form solution to \eqref{eq:reg_data_est} is given by:
\begin{align}\label{eq:closed_sol}
    \widehat{\mathbf{d}}_{\text{MMSE}} = (\overline{\mathbf{H}}_j^H\overline{\mathbf{H}}_j + \mu\mathbf{I})^{-1}\overline{\mathbf{H}}_j^H\mathbf{y}_j.
\end{align}
However, for large channel matrix $\overline{\mathbf{H}}_j$, it is complex to solve using \eqref{eq:closed_sol} due to computation of matrix inverse. Therefore, we consider gradient descent-based iterative method where gradient $\nabla J(\mathbf{d})$ in $t$th iteration and update rule for $(t+1)$th iteration are given below in \eqref{eq:gradient} and \eqref{eq:update_rule}.
\begin{subequations}
\begin{align}\label{eq:gradient}
    \nabla J(\mathbf{d})^{(t)} &= 2\overline{\mathbf{H}}_j^H(\overline{\mathbf{H}}_j\mathbf{d}^{(t)} - \mathbf{y}_j) + 2\mu\mathbf{d}^{(t)},\\\label{eq:update_rule}
    \mathbf{d}^{(t+1)} &= \mathbf{d}^{(t)} - \eta\nabla J(\mathbf{d})^{(t)},
\end{align}
\end{subequations}
where, $\eta$ denotes the step size used in the gradient descent updates. The inner-loop iteration for data estimation continues until convergence which is defined by the stopping criterion: $\lVert \mathbf{d}^{(t+1)} - \mathbf{d}^{(t)} \rVert < \epsilon$, where $\epsilon$ is a pre-defined threshold that controls the convergence precision. Once the convergence is achieved and the final estimate of the data vector is obtained as $\mathbf{d}'$, it is transformed back into the delay-Doppler domain via inverse vectorization: $\mathbf{D}' = \text{vec}^{-1}(\mathbf{d}')$. Subsequently, the estimated information component in the delay-Doppler domain is extracted by subtracting the known pilot matrix: $\mathbf{D}_I' = \mathbf{D}' - \mathbf{D}_P$.
To recover the transmitted information symbols, a symbol-wise demodulation is performed. Specifically, the final estimate of the information data $\widehat{\mathbf{D}}_I$ is obtained by solving the following nearest neighbor projection problem:
\begin{align}\label{eq:symb_demod}
\widehat{\mathbf{D}}_I = \arg\min_{\mathbf{D}_I \in \mathbb{A}^{M \times N}} \lVert \mathbf{D}_I - \mathbf{D}_I' \rVert_{F}^2,
\end{align}
where $\mathbb{A}$ denotes the QAM modulation alphabet. With both the pilot $\mathbf{D}_P$ and the newly estimated information data $\widehat{\mathbf{D}}_I$ now available at the receiver, a refined channel estimation can be carried out using the active sensing framework described in Section~\ref{sec:act_sense}. In this stage, the full delay-Doppler matrix $\widehat{\mathbf{D}} = \widehat{\mathbf{D}}_I + \mathbf{D}_P$ is treated as known input. Following the updated channel estimate, the information data is re-estimated using the same regularized MMSE method approximation, leveraging the improved channel knowledge. This forms the basis of an \textit{outer-loop iteration}, wherein the data and channel estimates are refined alternately. The process continues until convergence is achieved, ensuring that both the channel parameters and information data are jointly optimized in a passive sensing and communication. Note that the above procedure for the ISAC is used at all the $Z$ receivers. Further, using the estimated channel parameters, the range and radial velocities of the targets to the receivers can be calculated using \eqref{eq:est_tar_para}\footnote{The computational complexity of the ISAC is investigated later while the discussion of \textbf{Algorithm~\ref{algo3}}.}. Next, we determine the location and velocity of a target using its estimated ranges and velocities to the $Z+1$ receivers.

\vspace{-2mm}
\section{Targets Localization Via Cooperative Sensing}\label{sec:loc_n_sen}
\begin{figure}[!t]
	\centering  \includegraphics[width=2.7in]{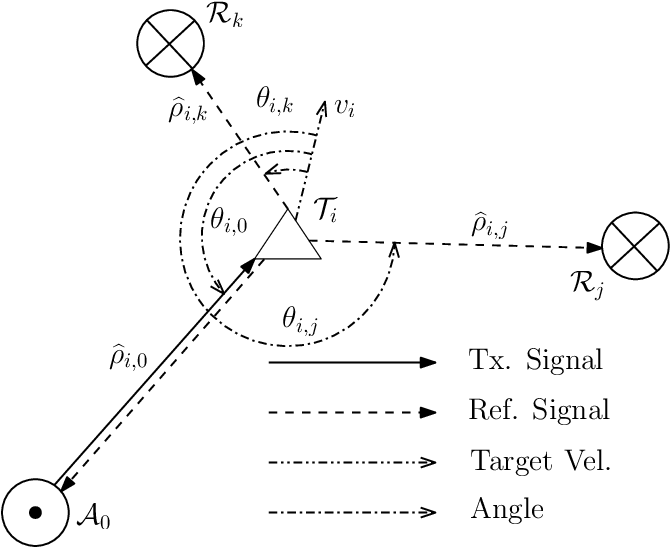}
	\caption{Estimation of location and velocity of a target via cooperative sensing.}    
	\label{fig:dist_vel_est}\vspace{-0mm}
\end{figure}

This section presents a method for determining the location and velocity of a target and then improving accuracy using the estimated ranges and radial velocities relative to the AN $\mathcal{A}_0$ and a set of spatially separated receivers. 

\vspace{-1mm}
\subsection{Targets' Location}
As illustrated in Fig.~\ref{fig:dist_vel_est}, location of a target $\mathcal{T}_i$ is estimated via geometric triangulation\footnote{Triangular geometry is adopted since three non-collinear nodes form the minimum configuration for unique 2D localization, and the corresponding localization error bound is inversely proportional to the triangle area. More complex layouts (e.g., square or hexagonal) reduce to multiple triangulations and offer no additional accuracy gain unless the effective triangular area is increased, while incurring higher coordination overhead.}
using transmitter $\mathcal{A}_0$ and two distinct receivers $\mathcal{R}_j$ and $\mathcal{R}_k$, where $j,k \in \{1,\ldots,Z\}$ and $j \neq k$ \cite{le21}. The estimated ranges $\widehat{\rho}_{i,0}$, $\widehat{\rho}_{i,j}$, and $\widehat{\rho}_{i,k}$ between the target and the three nodes are obtained using the parameter estimation procedure described in~\eqref{eq:est_tar_para}\footnote{Each receiver independently estimates the bistatic range and radial velocity using its locally received OTFS signal and forwards only these low-dimensional sensing parameters to the AN $\mathcal{A}_0$, which fuses them with its monostatic estimates for cooperative localization.}. Given the known coordinates of the transmitter and receivers, $(x_0,y_0)$ for $\mathcal{A}_0$ and $(x_j,y_j)$, $(x_k,y_k)$ for $\mathcal{R}_j$ and $\mathcal{R}_k$, the location $(\alpha_i,\beta_i)$ of the $i$th target is determined as follows.

First, Euclidean distance relations between the target and each node are expressed in~\eqref{eq:euclidean_dist}. By subtracting the equations corresponding to $q=j$ and $q=k$ from that of $q=0$, quadratic terms in the unknown target coordinates are eliminated, yielding two linear equations as shown in~\eqref{eq:simplified}. These equations are represented in matrix form in~\eqref{eq:matrix_form}.

\vspace{-0.5mm}
\begin{subequations}
\begin{align}\label{eq:euclidean_dist}
    &(x_q - \alpha_i)^2 + (y_q - \beta_i)^2 = \rho_{i,q}^2;\;\;\;\;\;\; q\in\{0,j,k\}\\\nonumber
    &\begin{bmatrix}
        x_q - x_0 & y_q - y_0
    \end{bmatrix}
    \begin{bmatrix}
        \alpha_i\\ 
        \beta_i
    \end{bmatrix}
    =\frac{1}{2}(\rho_{i,0}^2 - \rho_{i,q}^2 \\\label{eq:simplified}
    &\hspace{0.7in} - (x_0^2 - x_q^2) - (y_0^2 - y_q^2)); \;\;\;\; q \in \{j,k\}\\
    & \mathbf{A}
    \begin{bmatrix}
        \alpha_i\\ 
        \beta_i
    \end{bmatrix}
    = \mathbf{B};\;\;\;\; \mathbf{A} = 
    \begin{bmatrix}
        x_j - x_0 & y_j - y_0 \\\nonumber
        x_k - x_0 & y_j - y_0
    \end{bmatrix}\\\label{eq:matrix_form}
    & \mathbf{B} = \frac{1}{2}
    \begin{bmatrix}
        \rho_{i,0}^2 - \rho_{i,j}^2 - (x_0^2 - x_j^2) - (y_0^2 - y_j^2)) \\
        \rho_{i,0}^2 - \rho_{i,k}^2 - (x_0^2 - x_k^2) - (y_0^2 - y_k^2))
    \end{bmatrix}\\\label{eq:solution}
    & \begin{bmatrix}
        \alpha_i\\ 
        \beta_i
    \end{bmatrix}
    = \mathbf{A}^{-1}\mathbf{B}
\end{align}
\end{subequations}
The solution in~\eqref{eq:solution} uniquely determines location of the $i$th target provided the  matrix $\mathbf{A}$ is of  full rank, which is ensured when the transmitter and the two receivers are not collinear.

\vspace{-0.5mm}
\subsection{Targets' Velocity}
Using geometric triangulation, as illustrated in Fig.~\ref{fig:dist_vel_est}, we also estimate the velocity of the target. To do so, we first define the unit vectors $\{\mathbf{u}_{q,i}\}$ from the $i$th target to the two observing nodes $q \in \{j, k\}$. These are given by: $\mathbf{u}_{q,i} = [(x_q - \alpha_i)/\rho_{i,q}, (y_q - \beta_i)/\rho_{i,q}]^T$, where $(x_q, y_q)$ denotes the coordinates of node $q$, and $\rho_{i,q}$ is the Euclidean distance between the target and node $q$. Let the unknown velocity of the $i$th target be $\mathbf{v}_i = [v_{i,x}, v_{i,y}]^T$, where $v_{i,x}$ and $v_{i,y}$ are its components along the $x$ and $y$ axes, respectively. The projection of $\mathbf{v}_i$ onto the unit vector $\mathbf{u}_{q,i}$ yields the radial velocity measured by node $q$, denoted as $v_{i,q}$. This can be expressed as: $v_{i,q} = \mathbf{v}_i^T \mathbf{u}_{q,i}, \quad q \in \{j, k\}$. By arranging these equations in matrix form, we obtain a linear system:
\begin{align}\nonumber\label{eq:vel_cal}
    \!\!\!\!&\mathbf{C}
    \begin{bmatrix}
        v_{i,x} \\
        v_{i,y}
    \end{bmatrix}
    = \mathbf{D}
    \quad \Rightarrow \quad
    \begin{bmatrix}
        v_{i,x} \\
        v_{i,y}
    \end{bmatrix}
    = \mathbf{C}^{-1}\mathbf{D}, \;\; \text{where} \\
    \!\!\!\!&\mathbf{C} = 
    \begin{bmatrix}
        (x_j - \alpha_i)/\rho_{i,j} & (y_j - \beta_i)/\rho_{i,j} \\
        (x_k - \alpha_i)/\rho_{i,k} & (y_k - \beta_i)/\rho_{i,k}
    \end{bmatrix},\;
    \mathbf{D} = 
    \begin{bmatrix}
        v_{i,j} \\
        v_{i,k}
    \end{bmatrix}.
\end{align}
This system yields the solution for the velocity vector $\mathbf{v}_i$, assuming $\mathbf{C}$ is invertible, which is generally ensured by appropriate geometric placement of the nodes. Given $Z$ receivers, a total of $\mathcal{Z} \triangleq \binom{Z}{2}$ unique receiver pairs can be formed in conjunction with $\mathcal{A}_0$ to construct geometric triangles. Each of these $\mathcal{Z}$ configurations provides an independent estimate of a target's location and velocity using \eqref{eq:solution} and \eqref{eq:vel_cal}. The straightforward approach to reduce estimation error is to compute the arithmetic mean of the location and velocity estimates obtained from all $\mathcal{Z}$ triangle configurations. The averaged location and velocity of the $i$th target are given by
\begin{subequations}\label{eq:avg_approach}
\begin{align}\label{eq:avg_loc}
    (\alpha_{i,\text{avg}}, \beta_{i,\text{avg}}) &= \left( \frac{1}{\mathcal{Z}} \sum_{t=1}^{\mathcal{Z}} \alpha_{i,t},\; \frac{1}{\mathcal{Z}} \sum_{t=1}^{\mathcal{Z}} \beta_{i,t} \right), \\
    (v_{i,x,\text{avg}}, v_{i,y,\text{avg}}) &= \left( \frac{1}{\mathcal{Z}} \sum_{t=1}^{\mathcal{Z}} v_{i,x,t},\; \frac{1}{\mathcal{Z}} \sum_{t=1}^{\mathcal{Z}} v_{i,y,t} \right).
\end{align}
\end{subequations}
While this approach can mitigate random noise and improve robustness, it does not account for fact that certain triangle configurations may introduce higher errors due to unfavorable geometric conditions, e.g., near-collinear arrangements or poor baseline separations. To further enhance estimation accuracy by leveraging multiple geometric configurations, we next introduce a selection based approach that selectively utilizes only the most reliable component estimates.

\vspace{-0.5mm}
\subsection{Nearest-Neighbor-Based Selection Approach}
Averaging across all $\mathcal{Z}$ triangle-based estimates can help suppress random noise, but it does not consider the reliability of each configuration. In practice, some triangles may yield highly inaccurate estimates due to poor geometry, for example, when nodes are nearly collinear or too closely spaced. To address this, we propose a \textit{nearest-neighbor selection method} that prioritizes consistency among estimates. Let the set of location or velocity estimates for the $i$th target be represented by $\mathcal{C} = \{C_1, C_2, \dots, C_{\mathcal{Z}}\}$, where each $C_t \in \mathbb{R}^2$ corresponds to either $[\alpha_{i,t}, \beta_{i,t}]^T$ or $[v_{i,x,t}, v_{i,y,t}]^T$. To identify the most coherent subset of estimates, we calculate the Euclidean distance between each pair of points in $\mathcal{C}$. For a predefined threshold $\xi > 0$, the neighbor set $\mathcal{S}_n$ of an estimate $C_n$ includes all other estimates that lie within distance $\xi$ of it as
\begin{align}\label{eq:neigh_set}
	\mathcal{S}_n = \left\{ C_m \in \mathcal{C} \setminus \{C_n\} \;\middle|\; \lVert C_n - C_m \rVert \le \xi \right\}.
\end{align}
Among all such sets, we identify the one with the largest number of neighbors as
\begin{align}
	n^* = \arg \max_n \; \text{card}(\mathcal{S}_n).
\end{align}
The final estimate is then computed by averaging $C_{n^*}$ with all points in its neighbor set as
\begin{align}\label{eq:nearest_sel}
	C_{\text{avg}} = \frac{1}{\text{card}(\mathcal{S}_{n^*}) + 1} \left( C_{n^*} + \sum_{C_m \in \mathcal{S}_{n^*}} C_m \right).
\end{align}
By selecting and averaging only estimates that exhibit spatial coherence, this method reduces influence of unexpected estimates with large errors. Now, using estimated targets' location and velocity in \eqref{eq:avg_approach} or \eqref{eq:nearest_sel}, channel is estimated as follows.

\vspace{-0.5mm}
\subsection{Sensing-Aided Channel Reconstruction}
Using the refined estimates of the targets’ locations and velocities obtained in
\eqref{eq:avg_approach} or \eqref{eq:nearest_sel}, the underlying double-spread channel can be reconstructed in a geometry-consistent manner. This sensing-assisted approach exploits the known coordinates of the AN $\mathcal{A}_0$ and the receivers $\{\mathcal{R}_j\}_{j=1}^Z$ to recover the delay and Doppler parameters of each propagation path, thereby enabling low-dimensional yet accurate channel estimation. For the $i$th target, the estimated distance to node $q\in\{0,1,\ldots,Z\}$ is $\widehat{\rho}_{i,q}=\sqrt{(\widehat{\alpha}_i-x_q)^2+(\widehat{\beta}_i-y_q)^2}$. The corresponding propagation delays for the monostatic and bistatic links are
\begin{align}
    \!\!\!\!\!\widehat{\tau}_{i,0}\hspace{-0.5mm} = \hspace{-0.5mm}2\widehat{\rho}_{i,0}/c,\; \widehat{\tau}_{i,j}\hspace{-0.5mm} = \hspace{-0.5mm}(\widehat{\rho}_{i,j}\hspace{-0.5mm} + \hspace{-0.5mm}\widehat{\rho}_{i,0})/c;\; \text{for} \;\; j \hspace{-0.5mm}\in \hspace{-0.5mm}\{1,\cdots, Z\},\!\!
\end{align}
where $c$ denotes the speed of light. To reconstruct the Doppler shifts, define the radial unit vector from node $j$ to the
$i$th target as $\mathbf{u}_{i,j} = \left[(\widehat{\alpha}_{i}-x_j),\; (\widehat{\beta}_i - y_j)\right]^T/\widehat{\rho}_{i,j}$. Projecting the estimated velocity $\widehat{\mathbf{v}}_i=[\widehat{v}_{i,x},\widehat{v}_{i,y}]^T$
onto $\mathbf{u}_{i,j}$ yields the radial velocity $\widehat{v}_{i,j}=\widehat{\mathbf{v}}_i^T\mathbf{u}_{i,j}$. The resulting Doppler shifts are then obtained as
\begin{align}
    \widehat{\nu}_{i,0} = 2\widehat{v}_{i,0} f_c/c, \;\; \widehat{\nu}_{i,j} = \widehat{v}_{i,j}f_c/c  + \widehat{\nu}_{i,0}; \; \text{for}\; j\in \{1,\cdots,Z\},
\end{align}
where $f_c$ is the carrier frequency. With the reconstructed delays $\{\widehat{\tau}_{i,j}\}$ and Doppler shifts $\{\widehat{\nu}_{i,j}\}$, double-spread channel associated with each target--receiver pair is characterized. The corresponding channel gains $\{\widehat{h}_{i,j}\}$ are then estimated using \eqref{eq:h_est}. By leveraging the geometric and kinematic structure of the environment, this sensing-assisted channel reconstruction reduces estimation overhead while improving robustness compared with purely pilot-based approaches. Next, we examine how spatial placement of AN and receivers affects accuracy of the triangulation-based localization process.

\vspace{-0.5mm}
\subsection{Suboptimal Deployment for Estimation Error Minimization}
Using Fig.~\ref{fig:dist_vel_est}, we investigate the effect of the location of the AN and receivers on the estimation of the location of the targets using the following lemma\footnote{Due to space limitations, we focus on geometry-aware deployment for localization error minimization. A detailed analysis of \emph{variance reduction in target velocity estimation}, enabled by the proposed receiver deployment and multistatic OTFS-ISAC framework, will be investigated in future work.}.
\begin{lemma}\label{lemma1}
Consider the localization of the $i$th target $\mathcal{T}_i$ via geometrical triangulation using the AN $\mathcal{A}_0$ and two receivers $\mathcal{R}_j$ and $\mathcal{R}_k$. For given variances $\{\sigma_{i,q}^2\}$ of the squared range estimates $\{\widehat{\rho}_{i,q}^2| q \in \{0,j,k\}\}$ from the three nodes, the localization error covariance matrix of the estimated target position $(\widehat{\alpha}_{i},\widehat{\beta}_{i})$ satisfies that its trace $\operatorname{tr}\!\left(\operatorname{cov}(\widehat{\alpha}_{i},\widehat{\beta}_{i})\right)$ or its maximum eigenvalue $\kappa_{\max}\!\left(\operatorname{cov}(\widehat{\alpha}_{i},\widehat{\beta}_{i})\right)$ can be minimized by optimizing the area of the underlying triangle formed by the three nodes.
\end{lemma}
\begin{proof}
    Without loss of generality, we assume that the AN $\mathcal{A}_0$ is placed at the origin (with coordinate ($0,0$)) and the coordinates of $\mathcal{R}_j$ and $\mathcal{R}_k$ are $(x_j,y_j)$ and $(x_k,y_k)$.  Using \eqref{eq:euclidean_dist}, it can be shown that for given $\{\sigma_{i,q}^2\}$ for $q\in \{0,j,k\}$, the variances $\sigma_{\alpha_i}^2$ and $\sigma_{\beta_i}^2$ in the estimation of the x and y coordinates of the location of the $i$th target using the geometrical triangulation in Fig.~\ref{fig:dist_vel_est}, is given by
    \begin{align}\nonumber
        \sigma_{\gamma_i}^2 = \frac{\sigma_{i,j}^2(g_j+g_k)^2 + \sigma_{i,0}^2g_k^2 + \sigma_{i,k}^2g_j^2}{\mathcal{D}},
    \end{align}
    where $\mathcal{D} = (x_j y_k - x_k y_j)^2$, the variable $g$ is chosen as $g = y$ when computing $\sigma_{\alpha_i}^2$ (x-coordinate variance), and $g = x$ when computing $\sigma_{\beta_i}^2$ (y-coordinate variance). We assume that the errors in the estimation of $\widehat{\alpha}_i$ and $\widehat{\beta}_i$ are independent of each other. Therefore, the covariance matrix $\operatorname{cov}(\widehat{\alpha}_{i},\widehat{\beta}_{i})$ is diagonal matrix. Thus, $\operatorname{tr}(\operatorname{cov}(\widehat{\alpha}_{i},\widehat{\beta}_{i}))$ and $\kappa_{\max}(\operatorname{cov}(\widehat{\alpha}_{i},\widehat{\beta}_{i}))$ are expressed as:
    \begin{subequations}
    \begin{align}\nonumber
        &\operatorname{tr}(\operatorname{cov}(\widehat{\alpha}_{i},\widehat{\beta}_{i})) = \sigma_{\alpha_i}^2 + \sigma_{\beta_i}^2 = \sigma_{i,j}^2[(x_j+x_k)^2 + (y_j+y_k)^2] \\\label{eq:trace_cov}
        & + \sigma_{i,0}^2[x_k^2 + y_k^2] + \sigma_{i,k}^2[x_j^2 + y_j^2]/(x_j y_k - x_k y_j)^2 \\\nonumber
        & \kappa_{\max}(\operatorname{cov}(\widehat{\alpha}_{i},\widehat{\beta}_{i})) = \frac{1}{(x_j y_k - x_k y_j)^2}\max\{\sigma_{i,j}^2(x_j+x_k)^2 \\\label{eq:eigen_cov}
        & + \sigma_{i,0}^2x_k^2 + \sigma_{i,k}^2x_j^2,\;\; \sigma_{i,j}^2(y_j+y_k)^2+ \sigma_{i,0}^2y_k^2 + \sigma_{i,k}^2y_j^2\}
    \end{align}
    \end{subequations}
    It can be observed that the denominator of \eqref{eq:trace_cov} and \eqref{eq:eigen_cov} is proportional to the magnitude of the area of the triangle formed using the three nodes. Thus, it is desired to maximize the area of the triangle in a bounded region. On the other hand, their numerators depend on the magnitude of the x and y coordinates of the receivers. For a lesser magnitude of the x and y coordinates, the area of the triangle decreases. Thus, the area under the triangle needs to be optimized to minimize the trace or the maximum eigenvalue.  
\end{proof}
\noindent One of the sub-optimal solution for the \textbf{Lemma~\ref{lemma1}} can be obtained using the following corollary.
\begin{corollary}\label{coro1}
    For the given variances, $\{\sigma_{i,q}^2\}$; $q\in \{0,j,k\}$, if the receivers are placed at orthogonal axis in \textbf{Lemma~\ref{lemma1}}, then, $\operatorname{tr}(\operatorname{cov}(\widehat{\alpha}_{i},\widehat{\beta}_{i}))$ or $\kappa_{\max}(\operatorname{cov}(\widehat{\alpha}_{i},\widehat{\beta}_{i}))$ can be minimized by maximizing the area of the triangle formed using the three nodes in a given bounded region. 
\end{corollary}
\begin{proof}
    If the receivers $\mathcal{R}_j$ and $\mathcal{R}_k$ are placed at the x and y axes (orthogonal axes) of the Cartesian coordinates in \textbf{Lemma~\ref{lemma1}}, then their coordinates are given by: $(x_j,0)$ and $(0,y_k)$, respectively. After, substituting these coordinates in \eqref{eq:trace_cov} and \eqref{eq:eigen_cov}, we get
    \begin{subequations}\label{eq:trace_eigen_three}
    \begin{align}\label{eq:trace_sp}
        &\operatorname{tr}(\operatorname{cov}(\widehat{\alpha}_{i},\widehat{\beta}_{i})) = \frac{\sigma_{i,j}^2 + \sigma_{i,k}^2}{y_k^2} + \frac{\sigma_{i,j}^2 + \sigma_{i,0}^2}{x_j^2}, \\\label{eq:eigen_sp}
        &\kappa_{\max}(\operatorname{cov}(\widehat{\alpha}_{i},\widehat{\beta}_{i})) = \max\left\{\frac{\sigma_{i,j}^2 + \sigma_{i,0}^2}{x_j^2}, \frac{\sigma_{i,j}^2 + \sigma_{i,k}^2}{y_k^2}\right\}.
    \end{align}
     \end{subequations}
     From \eqref{eq:trace_sp} and \eqref{eq:eigen_sp}, it is evident that $\operatorname{tr}(\operatorname{cov}(\widehat{\alpha}_{i},\widehat{\beta}_{i}))$ and $\kappa_{\max}(\operatorname{cov}(\widehat{\alpha}_{i},\widehat{\beta}_{i}))$ can be minimized by maximizing the magnitude of $x_j$ and $y_k$ which is equivalent to maximizing the area of right angled triangle formed using the three nodes.
\end{proof}
\textbf{Lemma~\ref{lemma1}} analyzes localization accuracy for a single triangulation. 
We next generalize this result to multiple triangulations from a multistatic deployment. \textbf{Lemma~\ref{lemma2}} derives the corresponding trace and maximum eigenvalue of the covariance matrix when these estimates are fused.

\begin{lemma}\label{lemma2}
Consider the localization of the $i$th target using $\mathfrak{N}$ independent geometrical triangulations. Each triangulation consists of the AN $\mathcal{A}_0$ located at $(0,0)$ and two receivers $\mathcal{R}_{i,1,\ell}$ and $\mathcal{R}_{i,2,\ell}$ located at $(x_{1,\ell},y_{1,\ell})$ and $(x_{2,\ell},y_{2,\ell})$, respectively, for $\ell\in\{1,\ldots,\mathfrak{N}\}$. Let $\sigma^2_{i,0,\ell}$, $\sigma^2_{i,1,\ell}$, and $\sigma^2_{i,2,\ell}$ denote the variances of the squared range estimates from $\mathcal{A}_0$, $\mathcal{R}_{i,1,\ell}$, and $\mathcal{R}_{i,2,\ell}$ in the $\ell$th triangulation. If the estimate $(\widehat{\alpha}_i,\widehat{\beta}_i)$ is obtained by arithmetic averaging of the $\mathfrak{N}$ independent estimates, then the trace and the maximum eigenvalue of the resulting localization error covariance matrix are given by
\begin{subequations}
    \begin{align}\nonumber
        &\operatorname{tr}(\operatorname{cov}(\widehat{\alpha}_{i},\widehat{\beta}_{i}))\hspace{-0.5mm} =\hspace{-0.5mm} \frac{1}{\mathfrak{N}^2}\hspace{-1mm}\sum_{\ell=1}^{\mathfrak{N}}\hspace{-0.5mm}\left[\sigma_{i,1,\ell}^2\{(x_{1,\ell}\hspace{-0.5mm}+\hspace{-0.5mm}x_{2,\ell})^2\hspace{-0.5mm} + \hspace{-0.5mm}(y_{1,\ell}\hspace{-0.5mm}+\hspace{-0.5mm}y_{2,\ell})^2\}\right. \\\label{eq:trace_cov_mul}
        &\left. + \sigma_{i,0,\ell}^2(x_{2,\ell}^2 \hspace{-0.5mm}+ \hspace{-0.5mm}y_{2,\ell}^2)\hspace{-0.5mm} + \hspace{-0.5mm}\sigma_{i,k,\ell}^2(x_{1,\ell}^2\hspace{-0.5mm} + \hspace{-0.5mm}y_{1,\ell}^2)\right]\hspace{-1mm}/(x_{1,\ell} y_{2,\ell}\hspace{-0.5mm} -\hspace{-0.5mm} x_{2,\ell} y_{1,\ell})^2, \\\nonumber
        & \kappa_{\max}(\operatorname{cov}(\widehat{\alpha}_{i},\widehat{\beta}_{i})) = \frac{1}{\mathfrak{N}^2}\max\left\{\sum_{\ell=1}^{\mathfrak{N}}\left[\sigma_{i,1,\ell}^2(x_{1,\ell}+x_{2,\ell})^2 \right.\right.\\\nonumber\label{eq:eigen_cov_mul} 
        & + \sigma_{i,0,\ell}^2x_{2,\ell}^2
        \left.\left. + \sigma_{i,2,\ell}^2x_{1,\ell}^2 \right]/\mathcal{D}_\ell^2,\;\; \sum_{\ell=1}^{\mathfrak{N}}\left[\sigma_{i,1,\ell}^2(y_{1,\ell}+y_{2,\ell})^2+ \right.\right.\\
        &\sigma_{i,0,\ell}^2y_{2,\ell}^2  
         \left. \left. + \sigma_{i,2,\ell}^2y_{1,\ell}^2\right]/\mathcal{D}_\ell^2\right\},
    \end{align}      
\end{subequations}
where $\mathcal{D}_\ell \triangleq x_{1,\ell}y_{2,\ell}-x_{2,\ell}y_{1,\ell}$ denotes twice the signed area of the triangle formed by the AN and the two receivers in the $\ell$th triangulation. Here, $\mathfrak{N}=\mathcal{Z}$ for the averaging scheme in \eqref{eq:avg_loc} and $\mathfrak{N}=\text{card}(\mathcal{S}_{n^*})+1$ for the selection scheme in \eqref{eq:nearest_sel}. Moreover, both $\operatorname{tr}(\operatorname{cov}(\widehat{\alpha}_i,\widehat{\beta}_i))$ and $\kappa_{\max}(\operatorname{cov}(\widehat{\alpha}_i,\widehat{\beta}_i))$ are minimized by optimizing the area of each triangulation.
\end{lemma}
\begin{proof}
    The proof of \textbf{Lemma~\ref{lemma2}} can be carried out similarly to the proof of \textbf{Lemma~\ref{lemma1}}.
\end{proof}
Next, we show that the sub-optimal solution for \textbf{Lemma~\ref{lemma2}} is equivalent to place the two receivers with multiple uncorrelated antennas at the orthogonal axis with maximum underlying area.
\begin{lemma}\label{lemma3}
Consider the localization of the $i$th target using the AN $\mathcal{A}_0$ and $Z$ receivers with given variances of the squared range estimates. Suppose that the $Z$ receivers are equally partitioned into two groups and deployed along two orthogonal Cartesian axes. Then, the trace $\operatorname{tr}(\operatorname{cov}(\widehat{\alpha}_{i},\widehat{\beta}_{i}))$ and the maximum eigenvalue $\kappa_{\max}(\operatorname{cov}(\widehat{\alpha}_{i},\widehat{\beta}_{i}))$ of the resulting localization error covariance matrix are equivalent to those obtained in a topology consisting of the AN and two receivers, each equipped with $Z/2$ mutually uncorrelated antennas, placed on orthogonal axes. Moreover, in both topologies, these metrics are minimized by maximizing the the underlying triangulation area within a given bounded region.
\end{lemma}
\begin{proof}
Let the AN $\mathcal{A}_0$ be located at $(0,0)$ and the $Z$ receivers be equally divided into two groups of size $Z/2$, placed on the $x$- and $y$-axes at $(x_{1,\ell},0)$ and $(0,y_{2,\ell})$, respectively. Each localization estimate is obtained via a triangulation formed by $\mathcal{A}_0$, one receiver on the $x$-axis, and one on the $y$-axis. Hence, the total number of independent triangulations is $\mathfrak{N}=Z^2/4$. Using \textbf{Corollary~\ref{coro1}} of \textbf{Lemma~\ref{lemma1}} and the averaging result in \textbf{Lemma~\ref{lemma2}}, the trace and the maximum eigenvalue of the resulting localization error covariance matrix are given by
    \begin{subequations}
    \begin{align}\label{eq:trace_lemma3}
        &\operatorname{tr}(\operatorname{cov}(\widehat{\alpha}_{i},\widehat{\beta}_{i})) = \frac{4}{Z^2}\sum_{\ell = 1}^{Z^2/4}\left[\frac{\sigma_{i,1,\ell}^2 + \sigma_{i,0,\ell}^2}{x_{1,\ell}^2} + \frac{\sigma_{i,1,\ell}^2 + \sigma_{i,2,\ell}^2}{y_{2,\ell}^2} \right], \\\nonumber\label{eq:eigen_lemma3}
        &\kappa_{\max}(\operatorname{cov}(\widehat{\alpha}_{i},\widehat{\beta}_{i})) = \frac{4}{Z^2} \max \left\{\sum_{\ell = 1}^{Z^2/4}\frac{\sigma_{i,1,\ell}^2 + \sigma_{i,0,\ell}^2}{x_{1,\ell}^2}, \right. \\
        &\hspace{1.2in}\left. \sum_{\ell = 1}^{Z^2/4} \frac{\sigma_{i,1,\ell}^2 + \sigma_{i,2,\ell}^2}{y_{2,\ell}^2} \right\}.
    \end{align}
    \end{subequations}
For a given bounded deployment region, both \eqref{eq:trace_lemma3} and \eqref{eq:eigen_lemma3} are minimized by maximizing $|x_{1,t}|$ and $|y_{2,t}|$, which is equivalent to maximizing the area of each right-angled triangulation, as established in \textbf{Corollary~\ref{coro1}}. Since all receivers on each axis provide independent observations with identical geometry, their aggregate effect is equivalent to that of two receivers placed on orthogonal axes, each equipped with $Z/2$ mutually uncorrelated antennas. Each antenna pair yields independent triangulation with identical area and variance structure. Therefore, both deployment topologies yield same localization error covariance matrix, completing the proof.
\end{proof}

We next show that, with multiple independent antennas at each node and orthogonal receiver placement, the localization error covariance is reduced by a factor proportional to the cube of the number of antennas.
\begin{lemma}\label{lemma4}
    If each of the three nodes, the AN $\mathcal{A}_0$ and the receivers $\mathcal{R}_1$ and $\mathcal{R}_2$, has $N_a$ number of independent antennas (the separation between the two antennas at a node is greater than $\lambda/2$) and $\mathcal{R}_1$ and $\mathcal{R}_2$ are placed at orthogonal axes, then, the average trace $\operatorname{tr}(\operatorname{cov}(\widehat{\alpha}_{i,\text{avg}},\widehat{\beta}_{i,\text{avg}}))$  and the maximum eigenvalue $\kappa_{\max}(\operatorname{cov}(\widehat{\alpha}_{i,\text{avg}},\widehat{\beta}_{i,\text{avg}}))$, obtained using the approach in \eqref{eq:avg_loc}, are reduced by a factor $N_a^3$ compared to the obtained values of the counterparts in \eqref{eq:trace_eigen_three} of \textbf{Corollary~\ref{coro1}}. Here, we assume that the variances $\sigma_{i,0,\ell}^2$, $\sigma_{i,1,\ell}^2$, and $\sigma_{i,2,\ell}^2$ remain same for all $\ell\in\{1,\cdots,\mathfrak{N}\}$, where $\mathfrak{N} = N_a^3$.
\end{lemma}
\begin{proof}
Assume that each of the three nodes, namely the AN $\mathcal{A}_0$ and the receivers $\mathcal{R}_1$ and $\mathcal{R}_2$, is equipped with $N_a$ mutually uncorrelated antennas, where the inter-antenna spacing exceeds $\lambda/2$. Each antenna at $\mathcal{A}_0$ can be paired with any antenna at $\mathcal{R}_1$ and $\mathcal{R}_2$, which results in $\mathfrak{N}=N_a^3$ statistically independent geometrical triangulations for the $i$th target. Each triangulation provides an independent estimate of the squared ranges $\widehat{\rho}_{i,0}^2$, $\widehat{\rho}_{i,1}^2$, and $\widehat{\rho}_{i,2}^2$. Under the assumption that the per-antenna range estimation errors are independent and identically distributed with variances $\sigma_{i,0}^2$, $\sigma_{i,1}^2$, and $\sigma_{i,2}^2$, arithmetic averaging over the $N_a^3$ independent realizations yields effective variances that scale inversely with the number of triangulations, i.e., $\operatorname{var}(\widehat{\rho}_{i,q,\mathrm{avg}}^2)=\sigma_{i,q}^2/N_a^3$ for $q\in\{0,1,2\}$. This follows directly from standard variance reduction under independent averaging and is asymptotically justified by the central limit theorem. Substituting these reduced variances into the localization error expressions in \textbf{Corollary~\ref{coro1}}, and noting that both the trace and the maximum eigenvalue of the localization error covariance matrix depend linearly on the range estimation variances, it follows that $\operatorname{tr}(\operatorname{cov}(\widehat{\alpha}_{i,\mathrm{avg}},\widehat{\beta}_{i,\mathrm{avg}}))=\operatorname{tr}(\operatorname{cov}(\widehat{\alpha}_{i},\widehat{\beta}_{i}))/N_a^3$ and $\kappa_{\max}(\operatorname{cov}(\widehat{\alpha}_{i,\mathrm{avg}},\widehat{\beta}_{i,\mathrm{avg}}))=\kappa_{\max}(\operatorname{cov}(\widehat{\alpha}_{i},\widehat{\beta}_{i}))/N_a^3$, where $\operatorname{cov}(\widehat{\alpha}_{i},\widehat{\beta}_{i})$ corresponds to the single-antenna case given in \eqref{eq:trace_eigen_three}. Hence, when the receivers are placed along orthogonal axes, equipping each node with $N_a$ independent antennas reduces both the trace and the maximum eigenvalue of the localization error covariance matrix by a factor proportional to $N_a^3$, which completes the proof.
\end{proof}

Next, we introduce a KF framework that leverages the deployment-optimized sensing estimates to enable accurate and robust tracking of moving targets.

\vspace{-2mm}
\section{Sensing Improvement Using CRW-Based KF}
Here, we describe the trajectory estimation of moving targets using correlated random walk (CRW) based KF. Typically, the movements of the targets like humans, animals, vehicles, and uncrewed aerial vehicles (UAVs) are not entirely random, but they follow the same direction for a while. 

\vspace{-0.5mm}
\subsection{CRW-Based Prediction Model} CRW exploits the Ornstein-Uhlenbeck (OU) process to provide the prediction model \cite{joh02}. Using it, the velocity $\mathbf{v}_{i,t} = [v_{i,\alpha,t}, v_{i,\beta,t}]^T$ (with its $x$ and $y$ components $v_{i,\alpha,t}$ and $v_{i,\beta,t}$) of the  $i$th target at time $t$ is modeled as:
\begin{align}\label{eq:CRW}
    \mathrm{d}\mathbf{v}_{i,t} = \delta_i(\mathbf{v}_{i,t} - \bm{\omega}_{i})\mathrm{d}t + \psi_i \mathrm{d}\mathbf{W}_{i,t},
\end{align}
where $\delta_i$ is the velocity autocorrelation coefficient, $\bm{\omega}_i$ is the mean velocity, $\psi_i$ is the diffusion parameter and $\mathbf{W}_{i,t}$ is the Wiener random process. After solving \eqref{eq:CRW} in continuous time, we get $\mathbf{v}_{i,t}$ in \eqref{eq:vel_exp}, which is further represented in the discrete time domain as in \eqref{eq:vel_disc}.
\begin{subequations}
\begin{align}\label{eq:vel_exp}
    &\!\!\!\!\mathbf{v}_{i,t} = \bm{\omega}_i + (\mathbf{v}_{i,0} - \bm{\omega}_i)\mathrm{e}^{-\delta_i t} + \psi_i \int_{0}^t \mathrm{e}^{-\delta_i (t-s)}\mathrm{d}\mathbf{W}_{i,s},\!\!\!\! \\\label{eq:vel_disc}
    &\!\!\!\!\mathbf{v}_{i,t+1} = \mathbf{v}_{i,t}\mathrm{e}^{-\delta_i \Delta t} + \mathbf{n}_{i,t}^\mathbf{v},\!\!\!\!
\end{align}
\end{subequations}
where $\mathbf{v}_{i,0} = [v_{i,\alpha,0}, v_{i,\beta,0}]^T$ is the initial velocity of the target and $\mathbf{n}_{i,t}^v = [n_{i,t}^{v_\alpha}, n_{i,t}^{v_\beta}]^T$ is the Gaussian noise added to the velocity update. Using \eqref{eq:vel_exp}, the location $\mathbf{l}_{i,t} = [\alpha_{i,t}, \beta_{i,t}]$ of the target is updated as in \eqref{eq:loc_update}. Using it, the location update in discrete time, i.e., the location at $t+1$th time is updated from the location at $t$th time via \eqref{eq:loc_update_dis}.
\begin{subequations}
\begin{align}\label{eq:loc_update}
    &\mathbf{l}_{i,t} = \mathbf{l}_{i,0} + \int_{0}^{t} \mathbf{v}_{i,s} \mathrm{d}s, \\\label{eq:loc_update_dis}
    & \mathbf{l}_{i,t+1} = \mathbf{l}_{i,t} + \frac{1-\mathrm{e}^{-\delta_i \Delta t}}{\delta_i}\mathbf{v}_{i,t} + \mathbf{n}_{i,t}^\mathbf{l},
\end{align}
\end{subequations}
where $\mathbf{n}_{i,t}^\mathbf{l} = [n_{i,t}^\alpha, n_{i,t}^\beta]^T$ is Gaussian noise added to the location update. Next, we use the prediction model based on CRW in the KF algorithm to further enhance the accuracy in sensing of the targets. In this regard, first, we define the state $\mathbf{S}_t$ at time $t$ of the $i$th target as: $\mathbf{S}_t \triangleq [\alpha_{i,t}, \beta_{i,t}, v_{i,x,t}, v_{i,y,t}]^T$. Further, using \eqref{eq:vel_disc} and \eqref{eq:loc_update_dis}, the state transition is given in \eqref{eq:state_tran} and corresponding transition matrix $\mathbf{T}_i$ is expressed in \eqref{eq:tran_mat}, obtained by the Jacobian $J_f(\mathbf{s}_{i,t})$ of the state transition~\cite{ter03}.

\begin{subequations}
    \begin{align}\label{eq:state_tran}
        &\mathbf{s}_{i,t+1} = \mathbf{T}_i\mathbf{s}_{i,t} + \mathbf{n}_{i,t}^{\mathbf{s}}, \\\label{eq:tran_mat}
        &\mathbf{T}_i =
\begin{bmatrix}
1 & 0 & \frac{1 - e^{-\delta_i \Delta t}}{\delta_i} & 0 \\
0 & 1 & 0 & \frac{1 - e^{-\delta_i \Delta t}}{\delta_i} \\
0 & 0 & e^{-\delta_i \Delta t} & 0 \\
0 & 0 & 0 & e^{-\delta_i \Delta t},
\end{bmatrix}
    \end{align}
\end{subequations}
where $\mathbf{n}_{i,t}^{\mathbf{S}} = [n_{i,t}^{\alpha}, n_{i,t}^{\beta}, n_{i,t}^{v_\alpha}, n_{i,t}^{v_\beta}]^T$ is noise process. Noise $\mathbf{n}_{i,t}^{\mathbf{S}}$ captures uncertainties in prediction model via covariance matrix $\mathbf{Q}_{i,t} = \mathbb{E}[\mathbf{n}_{i,t}^{\mathbf{S}} {\mathbf{n}_{i,t}^{\mathbf{S}}}^T]$, with $\mathbb{E}[\mathbf{n}_{i,t}^{\mathbf{S}} {\mathbf{n}_{i,t^{'}}^{\mathbf{S}}}^T] = 0$ for $t\ne t^{'}$.

\vspace{-0.5mm}
\subsection{Measurement Model and Trajectory Estimation Using KF}
As described in Section~\ref{sec:loc_n_sen}, the measurement $\mathbf{z}_{i,t}$ for the location and velocity of the $i$th target is obtained from the average or selection based approaches using \eqref{eq:avg_approach} or \eqref{eq:nearest_sel}. The measurement $\mathbf{z}_{i,t}$ is related to the state $\mathbf{s}_{i,t}$ in the predication model, which is expressed as
\begin{align}\label{eq:meas_mod}
    \mathbf{z}_{i,t} = \mathbf{V}_i \mathbf{s}_{i,t} + \mathbf{n}_{i,t}^\mathbf{z},
\end{align}
where $\mathbf{V}_i$ is the observation matrix and $\mathbf{n}_{i,t}^{\mathbf{z}}$ is the zero-mean Gaussian noise with covariance matrix $\mathbf{R}_i = \mathbb{E}[\mathbf{n}_{i,t}^{\mathbf{z}}{\mathbf{n}_{i,t}^{\mathbf{z}}}^T]$ and $\mathbb{E}[\mathbf{n}_{i,t}^{\mathbf{z}}{\mathbf{n}_{i,t^{'}}^{\mathbf{z}}}^T] = 0$. As in the measurement, we consider all the four $x$ and $y$ components of the location and velocity of the target, therefore, $\mathbf{V}_i = \mathbf{I}_4$, where $\mathbf{I}_n$ is an identity matrix of size $n \times n$. Using \eqref{eq:state_tran} and \eqref{eq:meas_mod}, the KF performs the forecast using the prediction and data correction to enhance the accuracy in the estimation. The KF in \textbf{Algorithm~\ref{algo11}}\footnote{\textbf{Algorithm~\ref{algo11}} follows the standard KF recursion and is not claimed as a new contribution here. It is included for completeness, consistent notation, and to support clear understanding of the subsequent algorithms. This presentation allows a clear distinction between the classical KF operations and the newly proposed CRW-based modeling and geometry-aware OTFS-ISAC integration.} starts with the initialization of $\mathbf{T}_i$, $\mathbf{V}_i$, predication error covariance matrix $\mathbf{Q}_{i,0}$, mesurement $\mathbf{z}_{i,1}$ with noise covariance matrix $\mathbf{R}_{i,1}$ and an estimate of the initial state $\mathbf{s}^a_{i,0} = \mathbb{E}[\mathbf{s}_{i,0}]$ with error covariance matrix $\mathbf{P}_{i,0} = \mathbb{E}[(\mathbf{s}_{i,0} - \mathbf{s}^a_{i,0})(\mathbf{s}_{i,0} - \mathbf{s}^a_{i,0})^T]$, where $\mathbf{s}_{i,0} = [\alpha_{i,0}, \beta_{i,0}, v_{i,\alpha,0}, v_{i,\beta,0}]^T$. Thereafter, for the $t$th time step, $\mathbf{Q}_{i,t-1}$, $\mathbf{R}_{i,t}$, $\mathbf{s}^a_{i,t-1}$, $\mathbf{P}_{i,t-1}$, and $\mathbf{z}_{i,t}$ are inputted to the function $f_{KF}$ that determines $\mathbf{s}^f_{i,t}$, $\mathbf{P}^f_{i,t}$, $\mathbf{K}_{i,t}$, $\mathbf{s}^a_{i,t}$, and $\mathbf{P}_{i,t}$. Finally, the function provides the estimate of the state $\mathbf{s}^a_{i,t}$ and the covariance $\mathbf{P}_{i,t}$.

\begin{algorithm}[!t]
\caption{KF function $f_{KF}(\cdot)$ to estimate the state $\mathbf{s}_{i,t}^a$.}\label{algo11}
\begin{algorithmic}[1]
\State System parameters $\mathbf{T}_i$ and $\mathbf{V}_i$
\State \textbf{Input:} $\mathbf{Q}_{i,t-1}$, $\mathbf{R}_{i,t}$, $\mathbf{s}^a_{i,t-1}$, $\mathbf{P}_{i,t-1}$, and the measurement $\mathbf{z}_{i,t}$

\State Determine the forecasted state: $\mathbf{s}^f_{i,t} = T_i \mathbf{s}^a_{i,t-1}$\label{step:3}
\State Update the forecast error covariance: $\mathbf{P}^f_{i,t} = \mathbf{T}_i \mathbf{P}_{i,t-1} \mathbf{T}_i^T + \mathbf{Q}_{i,t-1}$\label{step:4}

\State Determine Kalman gain: $\mathbf{K}_{i,t} = \mathbf{P}^f_{i,t} \mathbf{V}_i^T (\mathbf{V}_i \mathbf{P}^f_{i,t} \mathbf{V}_i^T + \mathbf{R}_{i,t})^{-1}$\label{step:5}
\State Update state estimate: $\mathbf{s}^a_{i,t} = \mathbf{s}^f_{i,t} + \mathbf{K}_{i,t} (\mathbf{z}_{i,t} - \mathbf{V}_i \mathbf{s}^f_{i,t})$\label{step:6}
\State Update posterior error covariance: $\mathbf{P}_{i,t} = (\mathbf{I} - \mathbf{K}_{i,t} \mathbf{V}_i) \mathbf{P}^f_{i,t} (\mathbf{I} - \mathbf{K}_{i,t} \mathbf{V}_i)^T + \mathbf{K}_{i,t} \mathbf{R}_{i,t} \mathbf{K}_{i,t}^T$\label{step:7}

\State \textbf{Output:} State estimate $\mathbf{s}^a_{i,t}$ and the error covariance $\mathbf{P}_{i,t}$
\end{algorithmic}
\end{algorithm}

\vspace{-0.5mm}
\subsubsection*{Complexity of \textbf{Algorithm~\ref{algo11}}}
The computational complexity of one KF update in \textbf{Algorithm~\ref{algo11}} depends on the state dimension $n$ and the measurement dimension $m$. The prediction step incurs $\mathcal{O}(n^2)$ for state propagation and $\mathcal{O}(n^3)$ for covariance update. In the measurement update, Kalman gain computation and posterior covariance update require $\mathcal{O}(n^2 m + n m^2 + m^3)$ operations, while the state update costs $\mathcal{O}(nm)$. Overall, the complexity of a single KF update is $\mathcal{O}(n^3 + n^2 m + n m^2 + m^3)$, which reduces to $\mathcal{O}(n^3)$ when $m \leq n$.

\begin{algorithm}[!t]
\caption{Improved active sensing using KF.}\label{algo2}
\begin{algorithmic}[1]
\State \textbf{System parameters:} $\{\mathbf{T}_i\}$, $\{\mathbf{V}_i\}$ for $i \in \{1, \cdots, P\}$, and others independent parameters
\State \textbf{Input:} $\{\mathbf{s}_{i,0}^a\}$, $\{\mathbf{P}_{i,0}\}$
\For{each $t \in \{1,\cdots, \mathcal{T}\}$}
\For{each $i \in \{1,\cdots, P\}$}
\For{each $j \in \{0,\cdots, Z\}$}
\State Determine $(\widehat{\tau}_{i,j,t},\widehat{\nu}_{i,j,t})$ and $\widehat{h}_{i,j,t}$ from \eqref{eq:ith_est} using the known data $\mathbf{d}$\label{step:62}
\State Determine the range $\widehat{\rho}_{i,j,t}$ and the radial velocity $\widehat{v}_{i,j,t}$ from \eqref{eq:est_tar_para}\label{step:72}
\EndFor
\State Get $\mathbf{Q}_{i,t-1}$ and determine the averaged or selection based estimated observation vector $\mathbf{z}_{i,t} = [\widehat{\alpha}_{i,t},\widehat{\beta}_{i,t},\widehat{v}_{i,x,t},\widehat{v}_{i,y,t}]^T$ using \eqref{eq:avg_approach} or \eqref{eq:nearest_sel} and get the noise covariance $\mathbf{R}_{i,t}$\label{step:92}
\State Determine $\mathbf{s}_{i,t}^a$ and $\mathbf{P}_{i,t}$ using the function $f_{KF}(\mathbf{s}_{i,t-1}^a, \mathbf{P}_{i,t-1}$, $\mathbf{Q}(i,t-1)$, $\mathbf{R}_{i,t})$ in \textbf{Algorithm~\ref{algo11}}\label{step:102}
\EndFor
\EndFor
\end{algorithmic}
\end{algorithm}
In \textbf{Algorithm~\ref{algo2}}, the active sensing as described in Section~\ref{sec:act_sense} is further improved using the KF. Here, for each target $i$ and receiver $j$ at a given time $t$, the delay $\widehat{\tau}_{i,j}$, Doppler $\widehat{\nu}_{i,j}$, and the channel gain $\widehat{h}_{i,j}$ are estimated, followed by the range $\widehat{\rho}_{i,j}$ and the radial velocity $\widehat{v}_{i,j}$. Then, for each target $i$, the measurement $\mathbf{z}_{i,t}$ is determined; thereafter, using function $f_{KF}(\cdot)$ in \textbf{Algorithm~\ref{algo11}}, the sensing variables are estimated.

\vspace{-0.5mm}
\subsubsection*{Complexity of \textbf{Algorithm~\ref{algo2}}} 
The computational complexity of \textbf{Algorithm~\ref{algo2}} scales with the number of time steps $\mathcal{T}$, the number of targets $P$, the number of receivers $Z$, and the state dimension $n$. For each time instant and target, delay--Doppler parameter estimation at all $Z+1$ receivers requires correlation searches with complexity $\mathcal{O}((Z+1)G_\tau G_\nu MN)$, which dominates over the negligible cost of range--velocity mapping and measurement formation. The subsequent Kalman filter update incurs $\mathcal{O}(n^3)$ complexity. Accordingly, the overall complexity of \textbf{Algorithm~\ref{algo2}} is $\mathcal{O}\!\left(\mathcal{T} P \big((Z+1)G_\tau G_\nu MN + n^3\big)\right)$, indicating linear scaling with $\mathcal{T}$, $P$, and $Z$, quadratic dependence on the delay--Doppler grid resolution, and cubic dependence on the state dimension.

In \textbf{Algorithm~\ref{algo3}}, for each time $t$, target $i$ and receiver $j$, the coarse estimates of the delay, Doppler, and channel gain, $(\widehat{h}_{i,j,t}^{(0)}, \widehat{\tau}_{i,j,t}^{(0)}, \widehat{\nu}_{i,j,t}^{(0)})$, are first obtained using the pilot $\mathbf{d}_P$. These estimates are mapped to the range and velocity domain, from which the coarse location and velocity of each target are determined via the averaging or selection methods in \eqref{eq:avg_approach} and \eqref{eq:nearest_sel}. The reverse mapping from range and velocity is then applied to reconstruct channel parameters that serve as initialization for iterative refinement. The refinement is carried out in an \textit{outer loop} that alternates between communication data detection and sensing parameter update. Within the inner loop, the received data is iteratively refined using a gradient-based update (cf. \eqref{eq:gradient}, \eqref{eq:update_rule}) until convergence within the tolerance $\epsilon_d$. The recovered symbols are demodulated to obtain the information data, which are then used to improve the channel estimates $(\widehat{\mathbf{h}}_{j,t}, \widehat{\bm{\tau}}_{j,t}, \widehat{\bm{\nu}}_{j,t})$.  For each target $i$, the refined measurements are aggregated into the observation vector $\mathbf{z}_{i,t} = [\widehat{\alpha}_{i,t}, \widehat{\beta}_{i,t}, \widehat{v}_{i,x,t}, \widehat{v}_{i,y,t}]^T$. The KF function $f_{KF}$ from \textbf{Algorithm~\ref{algo11}} is then applied to update the state $\mathbf{s}_{i,t}^a$ and error covariance $\mathbf{P}_{i,t}$. This ensures reliable tracking of the location and velocity despite noise and estimation errors.  

The outer loop continues until the differences between successive channel estimates fall below the predefined thresholds $(\epsilon_h, \epsilon_{\tau}, \epsilon_{\nu})$. Once satisfied for all receivers, algorithm terminates after convergence. In this way, \textbf{Algorithm~\ref{algo3}} integrates communication data recovery and passive sensing in a closed-loop process, where improved data detection enhances channel estimation and, in turn, refined channel estimates improve sensing and tracking performance through KF.

\begin{breakablealgorithm}
\caption{Improved passive sensing and communication using KF.}
\label{algo3}
\begin{algorithmic}[1]
\State \textbf{System parameters:} $\{\mathbf{T}_i\}$ and $\{\mathbf{V}_i\}$ for $i \in \{1, \cdots, P\}$, thresholds $\epsilon_d$, $\epsilon_h$, $\epsilon_{\tau}$, $\epsilon_{\nu}$, and other independent parameters
\State \textbf{Input:} $\{\mathbf{s}_{i,0}^a\}$, $\{\mathbf{P}_{i,0}\}$
\For{each $t \in \{1,\cdots, \mathcal{T}\}$}
\For{each $i \in \{1, \cdots, P \}$}
\For{each $j \in \{0,\cdots, Z\}$}
\State Initially, estimate $(\widehat{h}_{i,j,t}^{(0)}, \widehat{\tau}_{i,j,t}^{(0)}, \widehat{\nu}_{i,j,t}^{(0)})$ via the coarse estimation using the pilot $\mathbf{d}_P$ in \eqref{eq:est_pilot_min} \label{step63}
\State Determine $\widehat{\rho}_{i,j,t}^{(0)}$ and $\widehat{\nu}_{i,j,t}^{(0)}$ using \eqref{eq:est_tar_para}\label{step73}
\EndFor
\State Determine the average or selection based estimated location $(\widehat{\alpha}_{i,t}^{(0)}, \widehat{\beta}_{i,t}^{(0)})$ and velocity $(\widehat{v}_{i,x,t}^{(0)}, \widehat{v}_{i,y,t}^{(0)})$ using \eqref{eq:avg_approach} or \eqref{eq:nearest_sel}\label{step93}
\For{each $j \in \{0,\cdots, Z\}$}
\State Using the reverse process of \eqref{eq:est_tar_para}, find the delay $\overline{\tau}_{i,j,t}^{(0)}$ and the Doppler shift $\overline{\nu}_{i,j,t}^{(0)}$
\State Using $(\overline{\tau}_{i,j,t}^{(0)}, \overline{\nu}_{i,j,t}^{(0)})$, obtain $\overline{h}_{i,j,t}^{(0)}$ for pilot $\mathbf{d}_P$ from \eqref{eq:h_est}
\EndFor\label{step133}
\EndFor
\State \textbf{Initialize outer loop:} $\text{iter}_{out} = 0$, $\text{conv}_{out} = \text{`FALSE'}$, $(\mathbf{h}_j^{curr}, \bm{\tau}_j^{curr}, \bm{\nu}_j^{curr}) = (\widehat{\mathbf{h}}_j^{(0)},\widehat{\bm{\tau}}_j^{(0)},\widehat{\bm{\nu}}_j^{(0)})$, and flags $(f_1,\cdots,f_P) = (\text{`FALSE'},\cdots,\text{`FALSE'})$
\While{NOT $\text{conv}_{out}$}
\
\For{each $j \in \{0,\cdots, Z\}$}
\State $\text{iter}_{out} = \text{iter}_{out} + 1$
\State \textbf{Data Estimation (inner loop):} 
\State Initialize $\mathbf{d}_{j,t}^{(0)} = \mathbf{0}$, $q = 0$, and $\text{conv}_{in} = \text{`FALSE'}$
\While{NOT $\text{conv}_{in}$}
\State Compute $\nabla J(\mathbf{d}_{j,t})^{(q)}$ using \eqref{eq:gradient}\label{step223}
\State Compute $\mathbf{d}_{j,t}^{(q+1)}$ using \eqref{eq:update_rule}
\If{$\|\mathbf{d}_{j,t}^{(q+1)} - \mathbf{d}_{j,t}^{(q)}\| < \epsilon_d$}
\State $\text{conv}_{in} = \text{`True'}$
\Else
\State $q = q + 1$\label{step273}
\EndIf
\EndWhile
\State Perform $\mathbf{d}_{j,t}^{'} = \mathbf{d}_{j,t}^{(q+1)}$ and $\mathbf{D}_{j,t}^{'} = \textbf{vec}^{-1}(\mathbf{d}_{j,t}^{'})$
\State Obtain information data $\mathbf{D}_{I,j,t}^{'} = \mathbf{D}_{j,t}^{'} - \mathbf{D}_P$
\State Obtain $\widehat{\mathbf{D}}_{I,j,t}$ from $\mathbf{D}_{j,t}^{'}$ using sybmol-wise demodulation from \eqref{eq:symb_demod}\label{step323}
\State \textbf{Refined channel estimation (active sensing):}
\State Perform $\widehat{\mathbf{D}}_{j,t} = \widehat{\mathbf{D}}_{I,j,t} + \mathbf{D}_P$ and $\widehat{\mathbf{d}}_{j,t} = \text{vec}(\widehat{\mathbf{D}}_{j,t})$
\If{$f_j == \text{`FALSE'}$}
\State Obtain $(\widehat{\mathbf{h}}_{j,t}, \widehat{\bm{\tau}}_{j,t}, \widehat{\bm{\nu}}_{j,t})$ from \eqref{eq:est_pilot_min} using the data $\widehat{\mathbf{d}}_{j,t}$\label{step363}
\State Determine the range $\widehat{\bm{\rho}}_{j,t}$ and the radial velocity $\widehat{\mathbf{v}}_{j,t}$ using \eqref{eq:est_tar_para}
\EndIf
\EndFor
\For{each $i\in \{1, \cdots, P\}$}
\State Get $\mathbf{Q}_{i,t-1}$ and determine the averaged or selection based estimated observation vector $\mathbf{z}_{i,t} = [\widehat{\alpha}_{i,t},\widehat{\beta}_{i,t},\widehat{v}_{i,x,t},\widehat{v}_{i,y,t}]^T$ using \eqref{eq:avg_approach} or \eqref{eq:nearest_sel} and get the noise covariance $\mathbf{R}_{i,t}$\label{step413}
\State Determine $\mathbf{s}_{i,t}^a$ and $\mathbf{P}_{i,t}$ using the function $f_{KF}(\mathbf{s}_{i,t-1}^a, \mathbf{P}_{i,t-1}$, $\mathbf{Q}_{i,t-1}$, $\mathbf{R}_{i,t})$ in \textbf{Algorithm~\ref{algo11}}\label{step423}
\EndFor
\For{each $j \in \{0,\cdots, Z\}$}
\If{$f_j == \text{`FALSE'}$}
\State Using the reverse process of \eqref{eq:est_tar_para}, find the delay $\bm{\tau}_{j,t}^{new}$ and the Doppler shift $\bm{\nu}_{j,t}^{new}$
\State Using $(\bm{\tau}_{j,t}^{new}, \bm{\nu}_{j,t}^{new})$, obtain $\mathbf{h}_{j,t}^{new}$ for data $\widehat{\mathbf{d}}_{j,t}$ from \eqref{eq:h_est}
\If{$\lVert \mathbf{h}_{j,t}^{new} - \mathbf{h}_{j,t}^{curr} \rVert < \epsilon_h$ \& $\lVert \bm{\tau}_{j,t}^{new} - \bm{\tau}_{j,t}^{curr} \rVert < \epsilon_{\tau}$ \& $\lVert \bm{\nu}_{j,t}^{new} - \bm{\nu}_{j,t}^{curr} \rVert < \epsilon_{\nu}$}\label{step483}
\State $f_j = \text{`TRUE'}$
\Else 
\State \hspace{-5mm}$(\mathbf{h}_{j,t}^{curr}, \hspace{-0.5mm}\bm{\tau}_{j,t}^{curr}, \hspace{-0.5mm}\bm{\nu}_{j,t}^{curr})\hspace{-0.5mm} = \hspace{-0.5mm}(\mathbf{h}_{j,t}^{new},\hspace{-0.5mm} \bm{\tau}_{j,t}^{new},\hspace{-0.5mm} \bm{\nu}_{j,t}^{new})$
\EndIf
\EndIf
\EndFor
\If{$(f_1,\cdots,f_P) = (\text{`TRUE'},\cdots,\text{`TRUE'})$}
\State $\text{conv}_{out} = \text{`TRUE'}$
\EndIf
\EndWhile
\EndFor
\end{algorithmic}
\end{breakablealgorithm}

\vspace{-0.5mm}
\subsubsection*{Complexity of \textbf{Algorithm~\ref{algo3}}}  
The computational complexity of \textbf{Algorithm~\ref{algo3}} scales with the number of time steps $\mathcal{T}$, the number of targets $P$, the number of receivers $Z$, the state dimension $n$, and the average numbers of inner and outer iterations, $I_{\mathrm{in}}$ and $I_{\mathrm{out}}$. For each time instant and target, the initial coarse channel estimation at all receivers requires delay--Doppler correlation searches with complexity $\mathcal{O}((Z+1)G_\tau G_\nu MN)$, which dominates over the negligible cost of range--velocity mappings. Within each outer iteration, communication data refinement via $I_{\mathrm{in}}$ gradient-based updates incurs $\mathcal{O}((Z+1)I_{\mathrm{in}}MN)$, followed by refined channel estimation with complexity $\mathcal{O}((Z+1)G_\tau G_\nu MN)$. The KF-based state update per target requires $\mathcal{O}(n^3)$, while convergence checks have negligible cost. Accounting for $I_{\mathrm{out}}$ outer iterations, the overall complexity of \textbf{Algorithm~\ref{algo3}} is
$\mathcal{O}\!\Big(\mathcal{T} P \big((Z+1)\big(G_\tau G_\nu MN + I_{\mathrm{out}}(I_{\mathrm{in}}MN + G_\tau G_\nu MN)\big) + n^3\big)\Big)$.

\vspace{-0.5mm}
\subsection{Future Directions in Equivalent Multi-Antenna Topologies for Beamforming and Interference Mitigation}

Due to space constraints, we briefly highlight future extensions enabled by the equivalent multi-antenna topology established in \textbf{Lemma~\ref{lemma4}}. This shows that the proposed orthogonal-axis deployment is equivalent, in terms of localization error scaling, to a three-node architecture in which the AN and two receivers are equipped with multiple mutually uncorrelated antennas, realizable via co-located elements with spacing $\geq \lambda/2$. This equivalence allows spatial degrees of freedom to be exploited alongside DD-domain sparsity, enabling beamforming gain and interference suppression without modifying the OTFS processing pipeline.

After OTFS demodulation, antenna-wise DD-domain observations can be coherently combined as
\begin{equation}
\tilde{z}_{j}(\tau,\nu)=\mathbf{w}_{j}^{H}(\tau,\nu)\mathbf{z}_{j}(\tau,\nu),
\end{equation}
where low-overhead combining schemes such as MRC, MVDR~\cite{hab22}, or MMSE can be applied using DD-domain gain or covariance estimates. Such spatial combining is well matched to OTFS-ISAC, where the DD-domain channel is quasi-static, and is particularly effective under high mobility. When sensing coexists with downlink communication, sparse DD-domain MIMO channel estimates further enable linear precoding or Tomlinson--Harashima-type pre-cancellation~\cite{wese23} for efficient multiuser interference mitigation. Existing DD-domain interference cancellation and regularized MMSE equalization can be strengthened by incorporating spatial filtering prior to DD-domain estimation. Finally, Kalman filter–based prediction of the target state $(\alpha,\beta,\mathbf{v})$ enables DD-domain search gating to a small predicted $(\tau,\nu)$ region, forming an efficient predict--focus--combine loop for dynamic vehicular scenarios.




\vspace{-2mm}
\section{Numerical Results}\label{sec:num_rel}
This section assesses the performance of the proposed framework through numerical simulations. Guided by \textbf{Lemma~\ref{lemma3}}, a geometry-aware deployment with three nodes forming the maximum-area triangle within a bounded region is adopted, as it yields superior sensing accuracy compared to randomly deployed single-antenna nodes. Unless stated otherwise, this topology is used throughout the simulations. OTFS modulation is employed with $M=256$ subcarriers and $N=16$ symbols, using $4$-QAM signaling, subcarrier spacing $\Delta f = 240$~kHz, carrier frequency $f_c = 30$~GHz, and a cyclic prefix of $64$ samples, operating at an SNR of $0$~dB~\cite{gau20}. The network comprises $N_B = 3$ sensing nodes and $N_T = 4$ moving targets within a $400 \times 400$~m$^2$ area, where each node is equipped with a single antenna unless specified otherwise. Target motion follows a CRW model with velocity autocorrelation coefficient $\delta_i = 1.5$, diffusion parameter $\psi_i = 0.5$, and sampling interval $\Delta t = 0.5$~s. Initial target positions, velocities, and headings are randomly initialized within the region. These trajectories serve solely as ground truth for evaluation, while all localization and tracking results are obtained exclusively using the proposed KF-assisted multistatic OTFS-ISAC framework.

\begin{figure}[!t]
	\centering\vspace{-3mm}\!\!\!\!
	\subfigure[]{\includegraphics[width=0.26\textwidth]{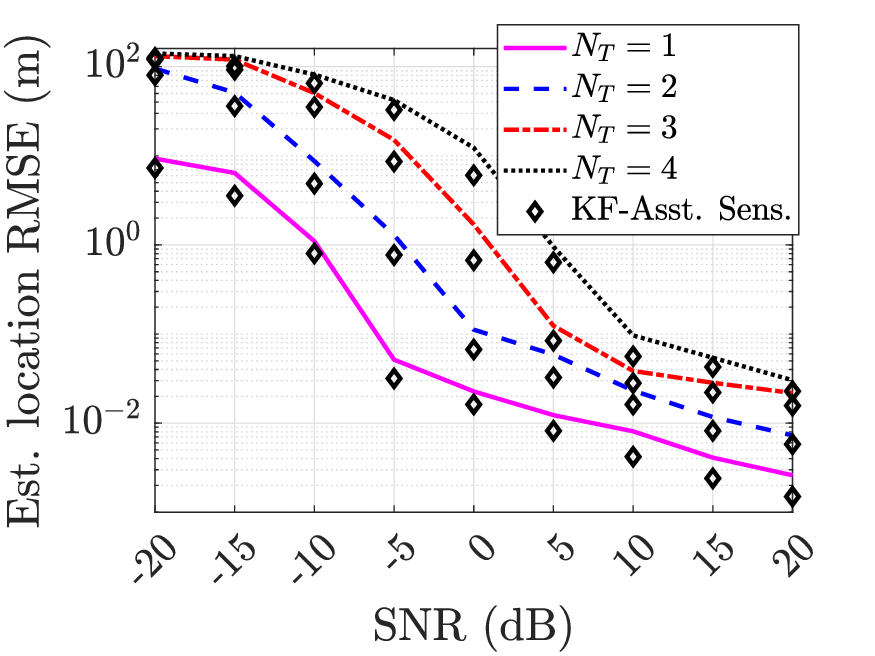}}\hspace{-2mm} \!\!\!\!
	\subfigure[]{\includegraphics[width=0.26\textwidth]{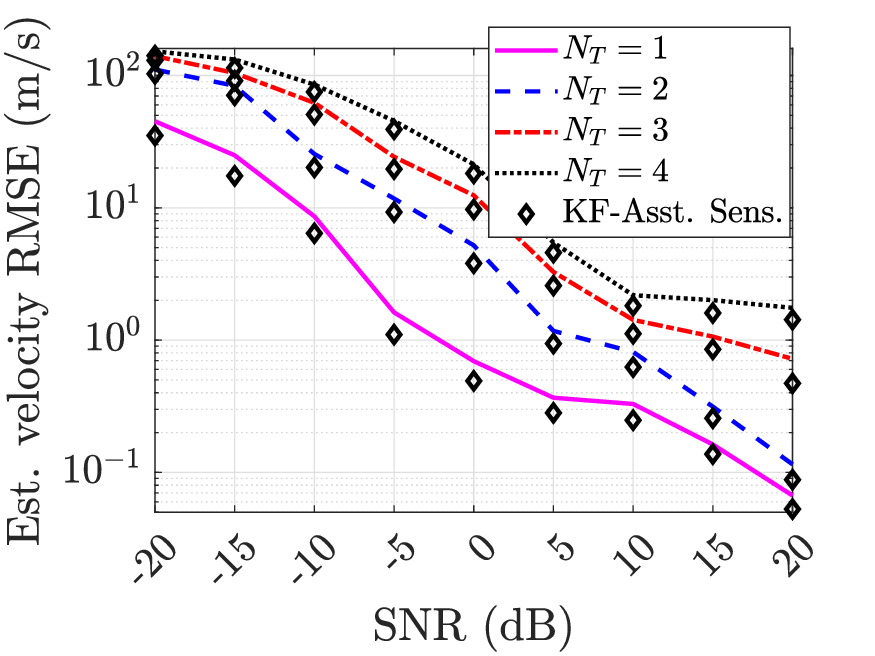}}\!\!\!\!
	\caption{RMSE improvement with SNR for different numbers of targets.}
	\label{fig:fig1}
\end{figure} 

\begin{figure}[!t]
	\centering\!\!\!\!\!\!
	\subfigure[]{\includegraphics[width=0.26\textwidth]{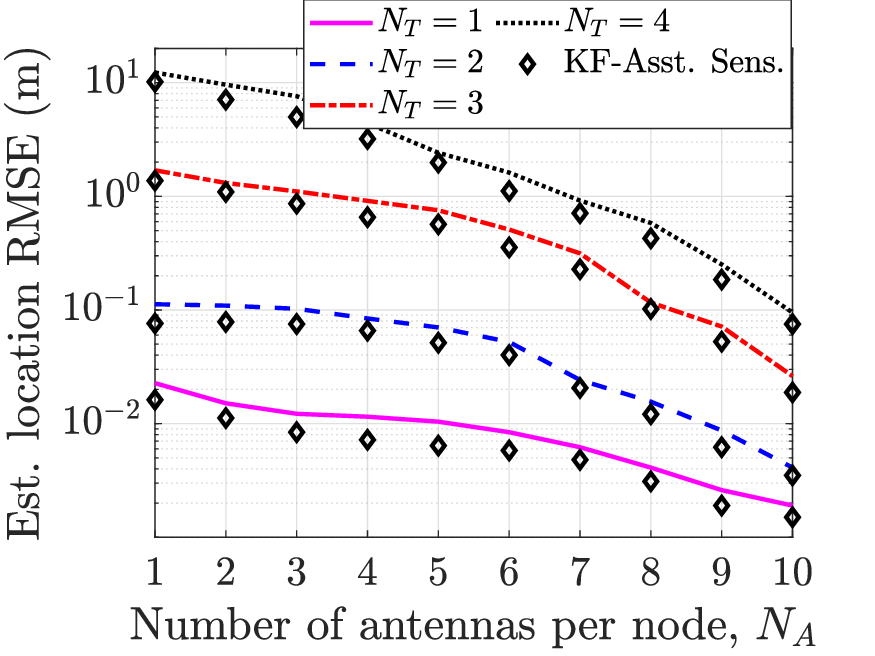}}\hspace{-2mm} \!\!\!\!\!\!
	\subfigure[]{\includegraphics[width=0.26\textwidth]{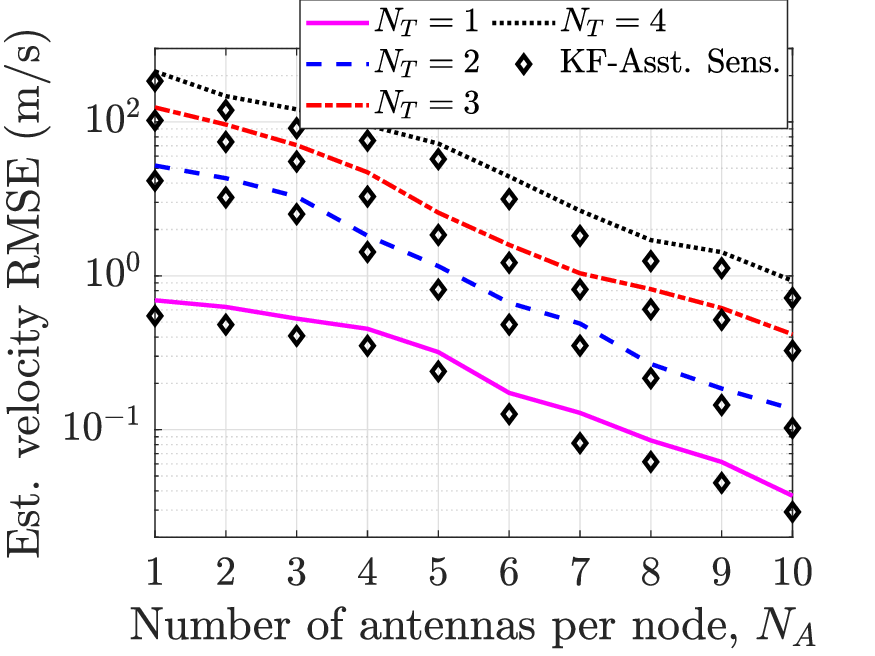}} \!\!\!\!\!\!
	\caption{RMSE improvement with number of antennas per node, $N_A$ for different numbers of targets, $N_T$.}
	\label{fig:fig2}
\end{figure} 

\begin{figure*}[!t]
	\centering\!\!\!\!
	\subfigure[]{\includegraphics[width=0.33\textwidth]{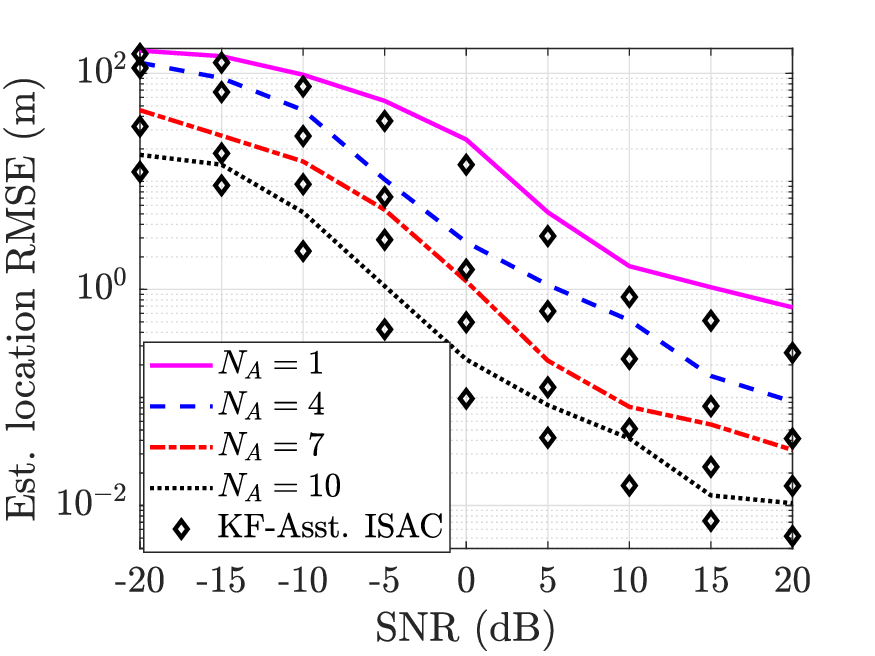}}
	\subfigure[]{\includegraphics[width=0.33\textwidth]{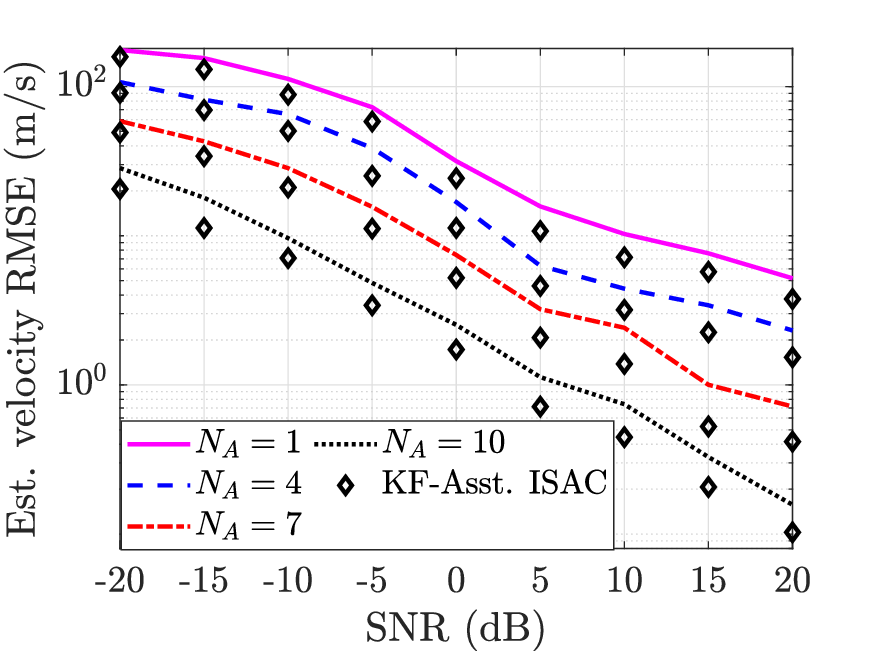}}
    \subfigure[]{\includegraphics[width=0.33\textwidth]{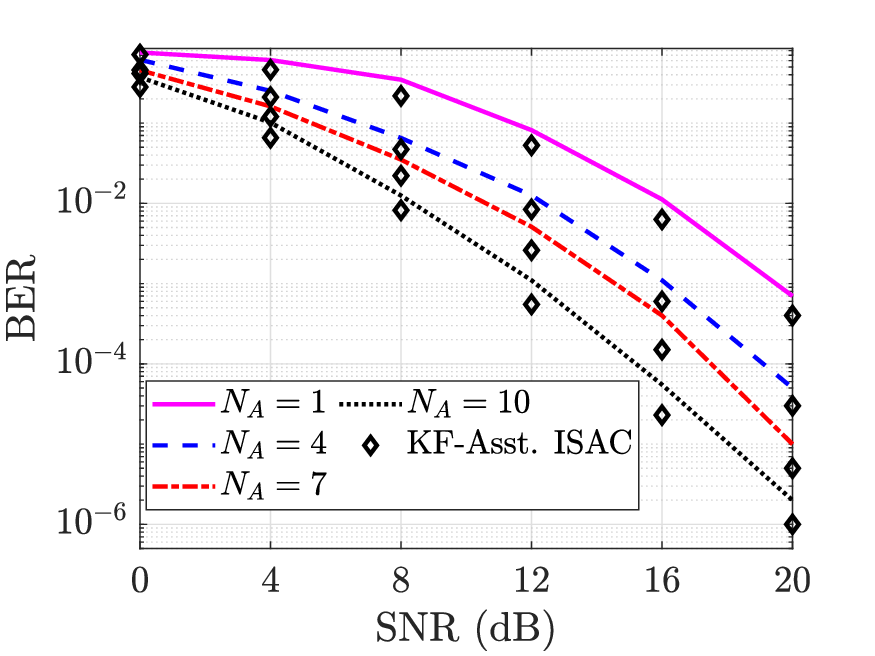}} \!\!\!\!
	\caption{Performance of ISAC in the localization and the data rate estimation.}
	\label{fig:fig3}
\end{figure*} 

Using Figs.~\ref{fig:fig1} and \ref{fig:fig2}, we evaluate active sensing and KF-assisted active sensing in terms of location and velocity RMSE. Across all SNRs and target counts, KF assistance consistently enhances performance by optimally fusing model-based prediction with noisy measurements. As shown in Fig.~\ref{fig:fig1}(a), the location RMSE versus SNR indicates markedly superior performance at $-20$~dB for a single target ($N_T=1$). For $N_T>1$, the RMSE increases by approximately $85$~m due to the growing difficulty of resolving multiple correlation peaks in low-SNR regimes (cf.~\eqref{eq:ith_est}). Averaged over SNR, the location RMSE increases by $15.53$~m, $18.00$~m, and $10.03$~m as $N_T$ increases from $1$ to $2$, $2$ to $3$, and $3$ to $4$, respectively, whereas KF-assisted sensing achieves an average RMSE reduction of $4.96$~m. Fig.~\ref{fig:fig1}(b) shows a similar trend for velocity estimation, where the average RMSE increases by $17.50$~m/s, $12.29$~m/s, and $10.87$~m/s for the same target increments, while KF assistance reduces the RMSE by $4.02$~m/s. Fig.~\ref{fig:fig2} further demonstrates that both location and velocity RMSE decrease monotonically with the number of antennas per node $N_A$ for all $N_T$, in agreement with \textbf{Lemma~\ref{lemma3}}, as $N_A^3$ independent observations yield cubic-order error reduction. For $N_A=1$, the average location and velocity RMSEs are $3.53$~m and $9.90$~m/s, respectively, which reduce to $1.18$~m and $3.44$~m/s when averaged over all $N_A$ and $N_T$, with additional KF-induced gains of $0.31$~m and $0.73$~m/s.

In Fig.~\ref{fig:fig3}, we evaluate the proposed ISAC framework in terms of localization RMSE, velocity RMSE, and achievable data rate versus SNR for different numbers of antennas per node $N_A$. At $0$~dB SNR, increasing $N_A$ from $1$ to $10$ substantially improves sensing accuracy for $N_T=4$ targets, reducing the location and velocity RMSE from $24.41$~m and $31.65$~m/s to $0.23$~m and $2.53$~m/s, respectively. Under the same setting in Fig.~\ref{fig:fig2}, active sensing further reduces the RMSE from $12.30$~m and $21.28$~m/s to $0.09$~m and $0.93$~m/s, confirming the superiority of active sensing over passive ISAC due to full waveform knowledge at the receiver, whereas passive ISAC relies on iterative signal recovery (\textbf{Algorithm~\ref{algo3}}). Averaged over SNR, the baseline ISAC achieves RMSEs of $25.06$~m and $31.63$~m/s in location and velocity, respectively, which are reduced to $19.76$~m and $25.46$~m/s with KF assistance. Moreover, as shown in Fig.~\ref{fig:fig3}(c), the BER decreases from approximately $0.30$ at $N_A=1$ to $8.13\times10^{-2}$ at $N_A=4$, with average BERs of $0.16$ for ISAC and $0.13$ for KF-assisted ISAC across the considered SNR range.

   \begin{figure}[!t]
		\centering\vspace{-3mm}\!\!\!\!
		\subfigure[]{\includegraphics[width=0.26\textwidth]{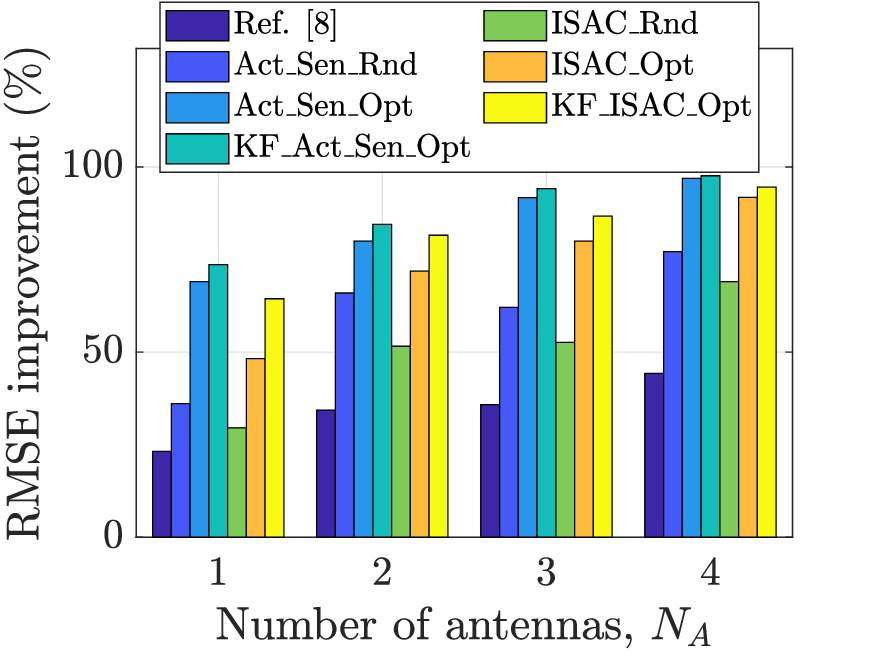}}\hspace{-2mm} \!\!\!\!
		\subfigure[]{\includegraphics[width=0.26\textwidth]{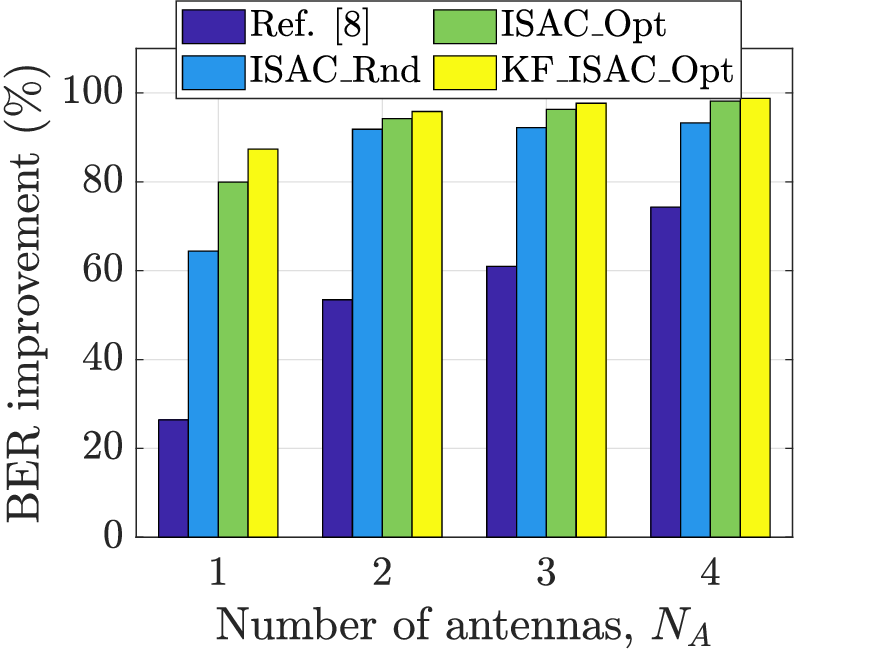}}\!\!\!\!
		\caption{Performance comparison of different schemes.}
		\label{fig:Scheme_Comp}\vspace{-5mm}
	\end{figure}

Fig.~\ref{fig:Scheme_Comp} compares the sensing RMSE and BER performance of the proposed multistatic OTFS-based schemes relative to the monostatic OTFS baseline under identical system parameters. The active sensing schemes include Act\_Sen\_Rnd (random deployment), Act\_Sen\_Opt (geometry-aware deployment), and KF\_Act\_Sen\_Opt (KF-assisted tracking with optimized deployment), while ISAC\_Rnd, ISAC\_Opt, and KF\_ISAC\_Opt denote their passive ISAC counterparts. As shown in Fig.~\ref{fig:Scheme_Comp}(a), all multistatic schemes outperform the monostatic baseline, demonstrating the fundamental gains from spatial diversity. While Act\_Sen\_Rnd improves RMSE through bistatic observations, its performance is constrained by non-ideal geometry. Act\_Sen\_Opt achieves substantially lower RMSE by maximizing the triangulation area, and KF\_Act\_Sen\_Opt further reduces RMSE by exploiting temporal correlation in target dynamics, yielding the best sensing performance across all antenna configurations. Similar trends are observed in ISAC: ISAC\_Rnd offers marginal gains, whereas ISAC\_Opt and KF\_ISAC\_Opt achieve pronounced RMSE reductions, with KF\_ISAC\_Opt consistently providing the lowest error, particularly at larger antenna counts. Fig.~\ref{fig:Scheme_Comp}(b) reports BER performance at an SNR of $8$~dB. For $N_A=1$, ISAC\_Rnd, ISAC\_Opt, and KF\_ISAC\_Opt reduce the BER from $1.7284$ to $0.6152$, $0.3472$, and $0.2182$, corresponding to improvements of $64.41\%$, $79.91\%$, and $87.37\%$, respectively. When the antenna count increases to $N_A=4$, the BER further decreases to $0.0920$, $0.0652$, and $0.0471$, yielding improvements of $93.96\%$, $98.14\%$, and $98.78\%$, respectively. The performance gap with increasing $N_A$ highlights the superior scalability of KF\_ISAC\_Opt. Overall, geometry-aware multistatic OTFS--ISAC with KF-assisted tracking provides the most significant gains over monostatic sensing and communication, making it well suited for multi-antenna vehicular networks.


\vspace{-2mm}
\section{Conclusion}
This work investigated a multistatic OTFS-based ISAC system for vehicular networks, addressing target localization, velocity estimation, and data communication under high-mobility conditions. By integrating cooperative triangulation with a CRW-based Kalman filtering framework, the system enables unified sensing, tracking, and communication using low-complexity nodes. The results reveal that localization accuracy is governed by triangulation area, and that orthogonal receiver deployment achieves error performance. Furthermore, an equivalent multi-antenna interpretation shows that increasing independent spatial observations yields cubic-order reduction in localization error. In joint passive sensing and communication, improved sensing fidelity enhances communication reliability, confirming the effectiveness of the multistatic OTFS-ISAC framework for vehicular environments.

\vspace{-2mm}

\makeatletter
\renewenvironment{thebibliography}[1]{%
	\@xp\section\@xp*\@xp{\refname}%
	\normalfont\footnotesize\labelsep .5em\relax
	\renewcommand\theenumiv{\arabic{enumiv}}\let\p@enumiv\@empty
	\vspace*{-1pt}
	\list{\@biblabel{\theenumiv}}{\settowidth\labelwidth{\@biblabel{#1}}%
		\leftmargin\labelwidth \advance\leftmargin\labelsep
		\usecounter{enumiv}}%
	\sloppy \clubpenalty\@M \widowpenalty\clubpenalty
	\sfcode`\.=\@m
}{%
	\def\@noitemerr{\@latex@warning{Empty `thebibliography' environment}}%
	\endlist
}
\makeatother     

\bibliography{references_COMML}

\end{document}